\documentclass{amsart}

\usepackage{amsmath}
\usepackage{amssymb}
\usepackage{amsthm}
\usepackage{color}
\usepackage[all]{xy}
\usepackage{enumerate}

\theoremstyle{definition}
\newtheorem{theorem}{Theorem}[section]
\newtheorem{lemma}[theorem]{Lemma}
\newtheorem{prop}[theorem]{Proposition}
\newtheorem{definition}[theorem]{Definition}

\newtheorem{cor}[theorem]{Corollary}

\theoremstyle{remark}
\newtheorem{rem}[theorem]{Remark}
\newtheorem*{acknowledgment}{Acknowledgments}

\DeclareMathOperator{\Q}{\mathbb{Q}}
\DeclareMathOperator{\N}{\mathbb{N}}
\DeclareMathOperator{\Z}{\mathbb{Z}} 

\DeclareMathOperator{\EE}{\mathbb{E}} 
\DeclareMathOperator{\ZZ}{\mathcal{Z}} 
\DeclareMathOperator{\A}{\mathbb{A}}

\DeclareMathOperator{\Ah}{\widehat{\mathbb{A}}}

\DeclareMathOperator{\id}{id} 
\DeclareMathOperator{\sgn}{sgn}

\newcommand{\seq}[1]{{\boldsymbol{#1}}}
\newcommand{\vev}[1]{{\langle#1\rangle}}
\newcommand{\bvev}[1]{{\bigl\langle#1\bigr\rangle}}

\title{Quantum dilogarithms and partition $q$-series}

\author{Akishi~Kato}%
\address{Graduate School of Mathematical Sciences, 
The University of Tokyo, 
3-8-1 Komaba, Meguro-ku, Tokyo 153-8914, Japan.}
\email{akishi@ms.u-tokyo.ac.jp}

\author{Yuji~Terashima}%
\address{Graduate School of Information Science and Engineering, 
Tokyo Institute of Technology, 
2-12-1 Ookayama, Meguro-ku, Tokyo 152-8550, Japan.}
\email{tera@is.titech.ac.jp}

\date{2014-11-19}

\numberwithin{equation}{section}

\allowdisplaybreaks[4]

\begin{document}

\begin{abstract}
 In our previous work \cite{KT2014}, we introduced the partition
 $q$-series for mutation loop $\gamma$ --- a loop in exchange quiver.
 In this paper, we show that for 
 certain class of mutation sequences, 
 called reddening sequences,
 graded version of partition $q$-series essentially coincides with the
 ordered product of quantum dilogarithm associated with each mutation;
 the partition $q$-series provides a state-sum description of
 combinatorial Donaldson-Thomas invariants introduced by B. Keller.
\end{abstract}

\maketitle


\section{Introduction}

Kontsevich and Soibelman's groundbreaking work
\cite{Kontsevich2008,Kontsevich2010} introduced some completely new
ideas and techniques into the BPS state counting problems in
physics. Their work as well as Nagao \cite{Nagao2011,Nagao2013} and
Reineke \cite{Reineke2010,Reineke2011} motivated Keller
\cite{Keller2011,Keller2012,Keller2013a} to study the product of quantum
dilogarithms along a quiver mutation sequence. He showed that it is
independent of the choice of a reddening mutation sequence and is an
important invariant of a quiver which he called the combinatorial
Donaldson-Thomas (DT) invariant.

In our previous work \cite{KT2014}, we introduced the partition
$q$-series for a mutation loop. A mutation loop is a mutation sequence
supplemented by a boundary condition which specifies how the vertices of
the initial and the final quiver are identified. One of our motivation
is to provide a solid mathematical foundation to extract an essential
information of the partition function of a 3-dimensional gauge
theory. In particular, we showed for a special sequence of a Dynkin
quiver or square product thereof, the partition $q$-series reproduce
so-called fermionic character formulas of certain modules associated
with affine Lie algebras.

In this paper, we analyze the relationship between partition $q$-series
and the combinatorial DT-invariants. For that purpose, we refined the
definition of our partition $q$-series by introducing a (noncommutative)
grading and making it sensitive to ``orientation'' (green or red) of
each mutation.

The main result of this paper is summarized as follows (see Theorem
\ref{thm:main} for more precise statement): for any reddening
sequence (= a mutation sequence for which the combinatorial DT invariant
is defined), the refined version of the partition $q$-series
coincides with the combinatorial DT invariant (up to involution
$q\leftrightarrow q^{-1}$). Therefore, the partition $q$-series provide
``state-sum'' description of combinatorial DT-invariants that are given
in ``operator formalism''.

The paper is organized as follows. In Section \ref{sec:museq},
we recall some basic concepts of quiver mutation sequences and
$c$-vectors.  In Section \ref{sec:Zdef}, we introduce the (refined
version of) the partition $q$-series $\ZZ(\gamma)$ for the mutation loop
$\gamma$.  In Section \ref{sec:backtracking} we show that the partition
$q$-series are invariant under insertion/deletion of
backtracking. Section \ref{comb-DT} summarizes the basic facts about
reddening sequences and combinatorial DT-invariants. Section
\ref{sec:Z-and-DT} is the main part of this paper; we prove that for any
reddening mutation sequence, the partition $q$-series essentially
coincides with the combinatorial DT invariant. The final section is
devoted to some explicit computation of partition $q$-series for various
type of quivers.
\begin{acknowledgment}
 We would like to thank H. Fuji, K. Hikami, A. Kuniba, R. Inoue,
 J. Suzuki, S. Terashima, O. Warnaar and M. Yamazaki for for helpful
 discussion. 
 We would also like to thank the anonymous referee for
 their valuable comments and suggestions to improve the
 quality of the paper.
 This work was partially supported by Japan Society for the
 Promotion of Science (JSPS), Grants-in-Aid for Scientific Research
 Grant (KAKENHI) Number 23654079 and 25400083. This paper is dedicated
 to the memory of Kentaro Nagao.
\end{acknowledgment}

\section{Quiver mutation sequences}
\label{sec:museq}

\subsection{Quivers and  mutation sequences}
\label{sec:notation}

A \emph{quiver} $Q=(Q_{0},Q_{1})$ is an oriented graph with the set
$Q_{0}$ of arrows and the set $Q_{1}$ of vertices. In this paper all
quivers are assumed to be finite connected oriented graphs without loops
or 2-cycles:
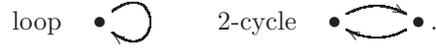
\begin{equation*}
 \text{loop} \quad \vcenter{
  \xymatrix @R=8mm @C=8mm
  @M=2pt{\bullet \ar@(ur,dr)[]}}  \qquad \qquad
  \text{2-cycle}\quad \vcenter{ \xymatrix @R=8mm @C=8mm @M=2pt{\bullet
  \ar@/^/[r] & \bullet \ar@/^/[l]} }.
\end{equation*}
Throughout the paper, we identify the set of vertices $Q_{0}$
with $\{1,2,\dots,n\}$. By a slight abuse of notation, we denote by
\begin{equation}
 \label{eq:Qij-def}
 Q_{ij}:=\#\{ (i\to j)\in Q_{1}\}
\end{equation}
the multiplicity of the arrow, and consider them as entries of an
$n\times n$ matrix.  There is a bijection
\begin{equation*}
\left\{
 \text{
 \parbox[c]{.35\textwidth}{the quivers without loops or $2$-cycles, 
 $Q_{0}\stackrel{\sim}{\rightarrow} \{1, \ldots, n\}$}
}
\right\}
\longleftrightarrow
\left\{
 \text{
 \parbox[c]{.35\textwidth}{the skew-symmetric integer $n\times
 n$-matrices $B$}
}
\right\}
\end{equation*}
via 
\begin{equation}
B_{ij}=Q_{ij}-Q_{ji},\qquad Q_{ij}=\max(B_{ij},0).
\end{equation}

For a quiver $Q$ and its vertex $k$, the \emph{mutated quiver}
$\mu_k(Q)$ is defined \cite{Fomin2002}: it has the same set of vertices
as $Q$; its set of arrows is obtained from that of $Q$ as follows:
\begin{itemize}
 \item[1)] for each path $i\to k\to j$ of length two, add a new arrow
	 $i\to j$;
 \item[2)] reverse all arrows with source or target $k$;
 \item[3)] remove the arrows in a maximal set of pairwise
	 disjoint $2$-cycles.
\end{itemize}
It is well known that $\mu_{k}(\mu_{k}(Q))=Q$ for any $1\leq k\leq n$.

A finite sequence of vertices of $Q$,
$\seq{m}=(m_{1},m_{2},\dots,m_{T})$ is called \emph{mutation
sequence}. By putting $Q(t) := \mu_{m_{t}}(Q(t-1))$, $\seq{m}$ induces a
(discrete) time evolution of quivers:
\begin{equation}
 \label{eq:Q-seq}
 \xygraph{!{<0cm,0cm>;<18mm,0cm>:<0cm,20pt>::}
!{(0,0)}*+{Q(0)}="q0"
!{(1,0)}*+{Q(1)}="q1"
!{(2,0)}*+{~\cdots~}="q2"
!{(2.9,0)}*+{Q(t{-}1)}="q3"
!{(4.1,0) }*+{Q(t)}="q4"
!{(5,0)}*+{~\cdots~}="q5"
!{(6,0) }*+{Q(T)}="q6"
"q0":^{\textstyle \mu_{m_{1}}}"q1"
"q1":^{\textstyle \mu_{m_{2}}}"q2"  
"q2":"q3"  
"q3":^(.56){\textstyle \mu_{m_{t}}}"q4"  
"q4":"q5"  
"q5":^{\textstyle \mu_{m_{T}}}"q6"  
}
\end{equation}
$Q(0)$ and $Q(T)$ are called the \emph{initial} and the \emph{final}
quiver, respectively. We will often use the notation
$\mu_{\seq{m}}(Q)=\mu_{m_{T}}(\cdots \mu_{m_{2}}(\mu_{m_{1}}(Q))\cdots)
$.

The quiver mutation corresponds to matrix mutation defined by
Fomin-Zelevinsky \cite{Fomin2002}.  The matrix $B(t)$ corresponding to
$Q(t)$ is given by \cite{Fomin2007}
\begin{equation} 
 \label{eq:matrix-mutation} B(t)_{ij} =
   \begin{cases}
    -B(t{-}1)_{ij} & 
    \text{if $i=k$ or $j=k$} \\
    B(t{-}1)_{ij}+\sgn(B(t{-}1)_{ik}) \max(B(t{-}1)_{ik}B(t{-}1)_{kj},0) 
    & \text{otherwise.} 
   \end{cases}
\end{equation}

Suppose that $Q(0)$ and $Q(T)$ are isomorphic as oriented graphs. An
isomorphism $\varphi : Q(T)\to Q(0)$ regarded as a bijection on the set
of vertices, is called \emph{boundary condition} of the mutation
sequence $\seq{m}$.  We represent $\varphi$ by a permutation of
$\{1,\dots,n\}$, i.e. $\varphi\in S_{n}$. The triple
$\gamma=(Q;\seq{m},\varphi)$ is called a \emph{mutation loop}.

\subsection{Ice quivers and $c$-vectors}
\label{sec:c-vectors}

We will follow the terminology in \cite{Bruestle2013}.  An \emph{ice
quiver} is a pair $(\widetilde{Q},F)$ where $\widetilde{Q}$ is a quiver
and $F\subset \widetilde{Q}_{0}$ is a (possibly empty) subset of
vertices called \emph{frozen vertices} such that there are no arrows
between them.  Two ice quivers $(\widetilde{Q},F)$ and
$(\widetilde{Q}',F')$ are called \emph{frozen isomorphic} if $F=F'$ and
there is an isomorphism of quivers $\phi : \widetilde{Q}\to
\widetilde{Q}'$ such that $\phi|F=\id_{F}$.

For any quiver $Q$, there is a standard way of constructing an ice
quiver  $Q^{\wedge}$ called \emph{framed quiver}.  $Q^{\wedge}$ is an
ice quiver obtained from $Q$ by adding, for each vertex $i$, a new
frozen vertex $i'$ and a new arrow $i\to i'$:
\begin{equation}
 F=\{i'|i\in Q_{0}\},\qquad  (Q^{\wedge})_{0}=Q_{0}\sqcup F,\qquad 
 (Q^{\wedge})_{1}=Q_{1}\sqcup \{i\to i'|i\in Q_{0}\}.
\end{equation}

Let $\seq{m}=(m_{1},m_{2},\dots,m_{T})$ be a mutation sequence for
$Q$. By putting
\begin{equation}
 \widetilde{Q}(0)=Q^{\wedge}, \qquad
  \widetilde{Q}(t)=\mu_{m_{t}}(\widetilde{Q}(t-1))\qquad (t=1,2,\dots,T)
\end{equation}
we can construct a sequence of ice quivers
\begin{equation}
 \label{eq:Q-tilde-seq}
 \xygraph{!{<0cm,0cm>;<18mm,0cm>:<0cm,20pt>::}
!{(0,0)}*+{\widetilde{Q}(0)}="q0"
!{(1,0)}*+{\widetilde{Q}(1)}="q1"
!{(2,0)}*+{~\cdots~}="q2"
!{(2.9,0)}*+{\widetilde{Q}(t{-}1)}="q3"
!{(4.1,0) }*+{\widetilde{Q}(t)}="q4"
!{(5,0)}*+{~\cdots~}="q5"
!{(6,0) }*+{\widetilde{Q}(T)}="q6"
"q0":^{\textstyle \mu_{m_{1}}}"q1"
"q1":^{\textstyle \mu_{m_{2}}}"q2"  
"q2":"q3"  
"q3":^{\textstyle \mu_{m_{t}}}"q4"  
"q4":"q5"  
"q5":^{\textstyle \mu_{m_{T}}}"q6"  
}.
\end{equation}
Note that we never mutate at frozen vertices $F=\{1',\dots,n'\}$.  The
quiver $\widetilde{Q}(t)$ will be called the \emph{ice quiver
corresponding to} $Q(t)$.  Let $\widetilde{B}(t)$ be the antisymmetric
matrix corresponding to $\widetilde{Q}(t)$.  The $c$-vectors are defined
by counting the number of arrows to/from frozen vertices:
\begin{definition}
 \label{def:c-vector} A \emph{$c$-vector of vertex $v$ in $Q(t)$} is a
 vector in $\Z^{n}$ defined by
 \begin{equation}
  c_{v}(t):=\bigl(\widetilde{B}(t)_{vi'}\bigr)_{i=1}^{n}.
 \end{equation}
\end{definition}
If the vertices of $\widetilde{Q}(t)$ are ordered as $(1,\dots,n,
1',\dots,n')$, the antisymmetric matrix $\widetilde{B}(t)$ has the
block form
\begin{equation}
 \label{eq:B-tilde-block}
 \def\h{\rule[-2.2ex]{0pt}{5.8ex}}
  \widetilde{B}(t) = 
  \begin{array}{|c|c|}
   \hline
    \h B(t) & ~C(t)~ \\
   \hline
    \h {-}C(t)^{\top} & 0 \\ \hline
  \end{array}\,,\qquad  C(t)=\begin{array}{|@{\qquad}c@{\qquad}|}
  \hline c_{1}(t) \\
  \hline c_{2}(t) \\
  \hline \cdots  \\
  \hline c_{n}(t) \\
  \hline 
      \end{array}\;,
\end{equation}
where $X^{\top}$ denotes the transpose of $X$.  The $n\times n$ block
$C(t)$ is called \emph{$c$-matrix}, which consists of row of
$c$-vectors.  By construction, $c_{i}(0)=e_{i}$, where $e_{i}$ is the
standard unit vector in $\Z^{n}$.

\begin{theorem}[Sign coherence]
 Each $c$-vector is nonzero and lies in $\N^{n}$ or $(-\N)^{n}$.
\end{theorem}
This is conjectured in \cite{Fomin2007} and was proved in
\cite{Derksen2010}, \cite{Plamondon2011}. Nagao \cite{Nagao2013} gave an
alternative proof by using Donaldson-Thomas theory.

\subsection{Green and red mutations}
\label{sec:green sequences}

Following Keller \cite{Keller2011}, we call the vertex $v$ of $Q(t)$ is
\emph{green} (resp. \emph{red}) if $c_{v}(t)\in \N^{n}$
(resp. $-c_{v}(t)\in \N^{n}$). By definition, every vertex of the
initial quiver $Q(0)$ is green. The mutation $ \mu_{m_{t}}:Q(t-1)\to
Q(t)$ is \emph{green} (resp. \emph{red}) if the mutating vertex $m_{t}$
is green (resp. red) on $Q(t-1)$, i.e. on the quiver before
mutation,\footnote{The sign of the mutating vertex changes after the
mutation. If the vertex $m_{t}$ is green on $Q(t-1)$, then it is red on
$Q(t)$.}  and the \emph{sign} $\varepsilon_{t}$ of the mutation
$\mu_{m_{t}}$ is defined as
\begin{equation}
 \label{eq:eps-def}
 \varepsilon_{t} = 
  \begin{cases}
   +1 & \text{if $\mu_{m_{t}}$ is green}, \\
   -1 & \text{if $\mu_{m_{t}}$ is red}. \\
  \end{cases}
\end{equation}

\begin{figure}[tb]
\begin{equation*}
 \xygraph{!{<0cm,0cm>;<16mm,0cm>:<0cm,23mm>::}
 !{(1,2)}*+{
  \vcenter{\xybox{\xygraph{!{<0cm,0cm>;<9mm,0cm>:<0cm,-9mm>::}
 !{(0,0)}*+<6.0pt>[Fo]{1}="v1"
 !{(1,0)}*+<6.0pt>[Fo]{2}="v2"
 !{(0,1)}*+<3pt>{1'}="w1"
 !{(1,1)}*+<3pt>{2'}="w2"
 "v1":"v2" 
 "v1":"w1" 
 "v2":"w2"  
 }}}}="Q1"
 !{(-0.3,1)}*+{
  \vcenter{\xybox{\xygraph{!{<0cm,0cm>;<9mm,0cm>:<0cm,-9mm>::}
 !{(0,0)}*+<3pt>[F]{1}="v1"
 !{(1,0)}*+<6.0pt>[Fo]{2}="v2"
 !{(0,1)}*+<3pt>{1'}="w1"
 !{(1,1)}*+<3pt>{2'}="w2"
 "v2":"v1" 
 "w1":"v1" 
 "v2":"w2"  
 }}}}="Q2"
 !{(1,0)}*+{
  \vcenter{\xybox{\xygraph{!{<0cm,0cm>;<9mm,0cm>:<0cm,-9mm>::}
 !{(0,0)}*+<3pt>[F]{1}="v1"
 !{(1,0)}*+<3pt>[F]{2}="v2"
 !{(0,1)}*+<3pt>{1'}="w1"
 !{(1,1)}*+<3pt>{2'}="w2"
 "v1":"v2" 
 "w1":"v1" 
 "w2":"v2"  
 }}}}="Q3"
 !{(3,2)}*+{
  \vcenter{\xybox{\xygraph{!{<0cm,0cm>;<9mm,0cm>:<0cm,-9mm>::}
 !{(0,0)}*+<6.0pt>[Fo]{1}="v1"
 !{(1,0)}*+<3pt>[F]{2}="v2"
 !{(0,1)}*+<3pt>{1'}="w1"
 !{(1,1)}*+<3pt>{2'}="w2"
 "v2":"v1" 
 "v1":"w1" 
 "v1":"w2" 
 "w2":"v2"  
 }}}}="Q4"
 !{(4.3,1)}*+{
  \vcenter{\xybox{\xygraph{!{<0cm,0cm>;<9mm,0cm>:<0cm,-9mm>::}
 !{(0,0)}*+<3pt>[F]{1}="v1"
 !{(1,0)}*+<6.0pt>[Fo]{2}="v2"
 !{(0,1)}*+<3pt>{1'}="w1"
 !{(1,1)}*+<3pt>{2'}="w2"
 "v1":"v2" 
 "v2":"w1" 
 "w2":"v1"  
 "w1":"v1"  
 }}}}="Q5"
 !{(3,0)}*+{
  \vcenter{\xybox{\xygraph{!{<0cm,0cm>;<9mm,0cm>:<0cm,-9mm>::}
 !{(0,0)}*+<3pt>[F]{1}="v1"
 !{(1,0)}*+<3pt>[F]{2}="v2"
 !{(0,1)}*+<3pt>{1'}="w1"
 !{(1,1)}*+<3pt>{2'}="w2"
 "v2":"v1" 
 "w2":"v1"  
 "w1":"v2"  
 }}}}="Q6"
 "Q1":@{=>}_{\textstyle \mu_{1}}"Q2"
 "Q2":@{=>}_{\textstyle \mu_{2}}"Q3"
 "Q1":@{=>}^{\textstyle \mu_{2}}"Q4"
 "Q4":@{=>}^{\textstyle \mu_{1}}"Q5"
 "Q5":@{=>}^{\textstyle \mu_{2}}"Q6"
 "Q3":@{<-}^{\textstyle (12)}_{\text{frozen isom.}}"Q6"
}
\end{equation*}
\caption{Pentagon and the $A_{2}$ quiver. The green and red vertices are
 marked with circles and boxes, respectively. Both 
 $\seq{m}=(1,2)$ and $\seq{m}'=(2,1,2)$ are maximal green sequences.}
 \label{fig:pentagon}
\end{figure}

A mutation sequence $\seq{m}=(m_{1},m_{2},\dots,m_{T})$ 
is called \emph{green sequence} if $m_{t}$
is green for all $t$, and is \emph{maximal green sequence} if all of the
vertex of the final quiver $Q(T)$ are red. 
In Figure \ref{fig:pentagon}, the two maximal green
sequences $(12)$ and $(212)$  are shown for $A_{2}$ quiver.

By inspecting the matrix mutation rules \eqref{eq:matrix-mutation} for
the ice quivers $\widetilde{Q}(t)$, it is easy to see how the
$c$-vector changes via mutations:
\begin{lemma}
 \label{prop:c-vec-change} 
 Under the mutation $ \mu_{v}:Q(t)\to
 Q(t+1)$, $c$-vector changes as
 \begin{equation}
  \label{eq:c-vec-change}
  c_{i}(t+1)=
   \begin{cases}
    -c_{i}(t)
    &  \text{if~} i=v 
    \\
    c_{i}(t)+ Q(t)_{i,v} \cdot c_{v}(t)  
    &
    \text{if~$i\neq v$ and $\mu_{v}$ is green} 
    \\
    c_{i}(t)+Q(t)_{v,i} \cdot  c_{v}(t)  
    & \text{if~$i\neq v$ and $\mu_{v}$ is red} 
   \end{cases}
 \end{equation}
\end{lemma}
\begin{cor}
 \label{prop:C(t)-as-basis} 
 $\det C(t)=(-1)^t$. In particular, $C(t)\in GL_{n}(\Z)$ and, $c$-vectors
 $\{c_{i}(t)\}_{i=1}^{n}$ constitutes a $\Z$-basis of $\Z^{n}$ for each $t$.
\end{cor}

\subsection{Noncommutative algebra $\Ah_{Q}$}

We introduce a noncommutative associative algebra in which quantum
dilogarithms and combinatorial Donaldson-Thomas invariants take their
values.

Let $Q$ be a quiver with vertices $\{1,2,\dots,n\}$. We define a skew
symmetric bilinear form $\vev{~,~}:\Z^{n}\times \Z^{n}\to \Z$ by
\begin{equation}
 \label{def:vev} 
  \vev{e_{i},e_{j}}=B_{ij}=-B_{ji}=Q_{ij}-Q_{ji},
\end{equation}
where $e_{1},\dots,e_{n}$ are the standard basis vectors in $\Z^{n}$.

Let $R$ be a commutative ring\footnote{The coefficient ring $R$ should
be chosen in such a way that the factors $q^{\pm \frac{1}{2}kk^{\vee}}$
of mutation weight \eqref{eq:W-def} belong to $R$. The exponent of $q$
can have nontrivial denominator through the process of expressing
$k^{\vee}$-variables in terms of $k$-variables.  As discussed in
\cite{KT2014}, there is a positive integer $\Delta$, depending only on
the mutation loop, such that $\frac{1}{2}kk^{\vee}\in
\frac{1}{\Delta}\Z$. Then we can choose $R=\Q(q^{1/\Delta})$.}
containing $\Q(q^{1/2})$.  Let
$\A_{Q}$ be a noncommutative associative algebra over $R$ presented as
\begin{equation}
 \label{eq:quantum-affine-space}
 \A_{Q}=R
  \langle\,
  y^{\alpha},~\alpha\in \N^{n} ~|~ 
  y^{\alpha}y^{\beta}
  = q^{\frac{1}{2}\vev{\alpha,\beta}}y^{\alpha+\beta}\,
  \rangle.
\end{equation}
Its completion with respect to the $\N^{n}$-grading is denoted by
$\widehat{\A}_{Q}$.  We may regard $\A_{Q}$ as the ring of
noncommutative polynomials in $y_{i}:=y^{e_{i}}$ $(i=1,\dots,n)$. Later
we will frequently use the following relations
($\alpha=(\alpha_{1},\dots,\alpha_{n})\in \Z^{n}$):
\begin{equation}	 
 \label{eq:y-monom}
\begin{split}
 & y_{1}^{\alpha_{1}}
  y_{2}^{\alpha_{2}}
  \dots 
  y_{n}^{\alpha_{n}}
  =q^{+\frac{1}{2}\sum_{i<j} B_{ij}\alpha_{i}\alpha_{j}}
  y^{\alpha},
 \\
 & y^{\alpha}=q^{-\frac{1}{2}\sum_{i<j} B_{ij}\alpha_{i}\alpha_{j}}
  y_{1}^{\alpha_{1}}
  y_{2}^{\alpha_{2}}
  \dots 
  y_{n}^{\alpha_{n}}.
\end{split}
\end{equation}

Later we will use a $\Q$-algebra anti-automorphism 
$\overline{\phantom{x}}:\A_{Q}\to \A_{Q}$ defined by
\begin{equation}
 \label{eq:anti-autom} y^{\alpha}\mapsto y^{\alpha}\quad (\alpha\in
   \N^{n}),\qquad q\mapsto q^{-1}
\end{equation}
Note that this is an involutive anti-automorphism of $\Q$-algebra, not
of $R$-algebra.

\section{Partition $q$-series}
\label{sec:Zdef}

In this section, we recapitulate the notion of partition $q$-series
introduced in \cite{KT2014}.  As mentioned in Introduction, we refine
and extend the definition, so that we can state the relationship between
our partition $q$-series and the products of quantum dilogarithm
(combinatorial DT-invariants) in full generality.
\begin{enumerate}[(i)]
 \itemsep=1ex
 \item We introduce noncommutative variables $y_{1},\dots,y_{n}$ to keep
      track of the $\N^{n}$-grading. This is in conformity to the custom
      of quantum dilogarithms and DT-invariants. They are naturally
      associated with the $c$-vectors as well as $s$-variables. This
      graded version of partition $q$-series now take their values in
      $\Ah_{Q}$ --- the (completed) ring of noncommutative polynomials
      in $y_{1},\dots,y_{n}$, rather than $\N[[q^{1/\Delta}]]$.

 \item We make a distinction between green and red mutations, and we add
      a new rule for red mutations.\footnote{Except the grading, all the
      results in \cite{KT2014} remains the same in our new setting; all
      the mutation sequence considered there are green sequences.}
      Although this refinement requires additional data --- $c$-vectors,
      or equivalently ice quivers, we obtain perfect match (Theorem
      \ref{thm:main}) between the partition $q$-series and the
      combinatorial Donaldson-Thomas invariant wherever the latter
      invariant are defined.

\end{enumerate}  
As a bonus of these refinements, we can handle arbitrary non-degenerate
mutation sequences\footnote{In \cite{KT2014}, the partition $q$-series
were well-defined only for ``positive'' mutation loops.}. Moreover, the
refined version acquire the invariance under the insertion or deletion
of backtracking in mutation sequence (Theorem
\ref{prop:Z-backtrack-inv}).

In the case of green mutation sequences, this new definition coincides
with the original one \cite{KT2014} just by forgetting $\N^{n}$
gradings.

\subsection{The partition $q$-series}

Let $Q$ be a quiver with vertices $\{1,2,\dots,n\}$.  We consider the
mutation sequence $\seq{m}=(m_{1},m_{2},\dots,m_{T})$ of $Q$:
\begin{equation}
 \label{eq:Q-seq-2}
 \xygraph{!{<0cm,0cm>;<18mm,0cm>:<0cm,20pt>::}
!{(0,0)}*+{Q(0)}="q0"
!{(1,0)}*+{Q(1)}="q1"
!{(2,0)}*+{~\cdots~}="q2"
!{(2.9,0)}*+{Q(t{-}1)}="q3"
!{(4.1,0) }*+{Q(t)}="q4"
!{(5,0)}*+{~\cdots~}="q5"
!{(6,0) }*+{Q(T)}="q6"
"q0":^{\textstyle \mu_{m_{1}}}"q1"
"q1":^{\textstyle \mu_{m_{2}}}"q2"  
"q2":"q3"  
"q3":^{\textstyle \mu_{m_{t}}}"q4"  
"q4":"q5"  
"q5":^{\textstyle \mu_{m_{T}}}"q6"  
}.
\end{equation}

The partition $q$-series is defined as follows \cite{KT2014}.
We first introduce a family of \emph{$s$-variables} $\{s_{i}\}$,
\emph{$k$-variables} $\{k_{t}\}$, and 
\emph{$k^{\vee}$-variables} $\{k^{\vee}_{t}\}$ by the following rule:
\begin{itemize}
 \item[(i)] An ``initial'' $s$-variable $s_{v}$ is attached to each
	    vertex $v$ of the initial quiver $Q$.

 \item[(ii)] Every time we mutate at vertex $v$, we add a ``new''
	    $s$-variable associated with $v$. We often write
	     $s_{v}$, $s_{v}'$, $s_{v}''$, $\dots$ to distinguish
	    $s$-variables attached to the same vertex.

 \item[(iii)] We associate $k_{t}$ and $k^{\vee}_{t}$ with each mutation
	    at $m_{t}$.
	      
 \item[(iv)] If two vertices are identified by a boundary condition, then
	    the corresponding $s$-variables are also identified.
\end{itemize}

The $s$-, $k$-, and $k^{\vee}$-variables are not considered
independent; we impose a linear relation for each mutation step.
Suppose that the quiver $Q(t-1)$ equipped with $s$-variables $\{s_{i}\}$
is mutated at vertex $v=m_{t}$ to give $Q(t)$. Then $k$- and
$s$-variables are required to satisfy
\begin{equation}
 \label{eq:k-s-rel}
 k_{t}=\begin{cases}
	\displaystyle
	s_{v}+s_{v}'-\sum_{a\to v} s_{a} & \text{if $\mu_{v}$ is
	green} ~(\varepsilon_{t}=1)\\
	\displaystyle \rule{0pt}{15pt}
	\sum_{v\to b} s_{b} - (s_{v}+s_{v}') & \text{if $\mu_{v}$ is
	red} ~(\varepsilon_{t}=-1)\\
       \end{cases}
\end{equation}
Here, $s'_{v}$ is the ``new'' $s$-variable attached to mutated vertex
$v$, and the sum is over all the arrows of $Q(t-1)$.
Similarly, $k^{\vee}$- and $s$-variables are related as
\begin{equation}
 \label{eq:kv-s-rel}
 k^{\vee}_{t}
  =\begin{cases}
    \displaystyle
	       s_{v}+s_{v}'-\sum_{v\to b} s_{b} & \text{if $\mu_{v}$ is
	green} ~(\varepsilon_{t}=1)\\
	\displaystyle \rule{0pt}{15pt}
    \sum_{a\to v} s_{a} - (s_{v}+s_{v}')
 & \text{if $\mu_{v}$ is
	red} ~(\varepsilon_{t}=-1)\\
   \end{cases}
\end{equation}
Therefore, 
\begin{equation}
 \label{eq:kv-k}
  k^{\vee}_{t}-k_{t}= \sum_{a\to v} s_{a} - \sum_{v\to b} s_{b}
\end{equation}
holds at each mutation.

The \emph{weight of the mutation} $\mu_{m_{t}}:Q(t-1)\to Q(t)$ at
$v=m_{t}$ is defined as
\begin{equation}
 \label{eq:W-def}
  W(m_{t}):=\frac{q^{\frac{\varepsilon_{t}}{2}k_{t}k^{\vee}_{t}}}
  {(q^{\varepsilon_{t}})_{k_{t}}} 
  =
  \begin{cases}
 \displaystyle
 \frac{q^{\frac{1}{2}k_{t}k^{\vee}_{t}}}
  {(q)_{k_{t}}} 
 & \text{if  $\mu_{v}$ is green},
 \\
 \displaystyle \rule{0pt}{22pt}
 \frac{q^{-\frac{1}{2}k_{t}k^{\vee}_{t}}}
  {(q^{-1})_{k_{t}}} 
 & \text{if  $\mu_{v}$ is red}.
  \end{cases}
\end{equation}
Here, $\varepsilon_{t}$ is the sign of $\mu_{m_{t}}$ and 
\begin{equation}
  (q)_{k} := \prod_{i=1}^{k}(1-q^{i})
\end{equation}
is the $q$-Pochhammer symbol. 

The \emph{$\N^{n}$-grading} of the mutation $\mu_{m_{t}}$ is
$k_{t} \alpha_{t}$ by definition, where
\begin{equation}
  \label{eq:alpha-def}
 \alpha_{t}:=\varepsilon_{t}c_{m_{t}}(t{-}1)\quad \in \N^{n}\setminus\{0\}
\end{equation}
is the (sign-corrected) $c$-vector of the vertex on which mutation is applied.

It is occasionally useful to regard the relation \eqref{eq:k-s-rel} as
the time evolution of $s$-variables with the control parameters
$\{k_{t}\}$.  Let $s_{i}(t)$ denote the value of the $s$-variable
associated with vertex $i$ at $Q(t)$. Then \eqref{eq:k-s-rel} can be
written as
\begin{equation}
 \label{eq:s-evolution}
 s_{i}(t)=
  \begin{cases}
   \rule{0pt}{14pt} s_{i}(t{-}1) & \text{if $i\neq v$},
   \\
   \displaystyle \rule{0pt}{19pt}
   k_{t}-s_{v}(t{-}1)+\sum_{a}Q(t)_{a,v} s_{a}(t{-}1)
   & \text{if $i=v$ and $\mu_{v}$ is green},
   \\
   \displaystyle \rule{0pt}{13pt}
   -k_{t}-s_{v}(t{-}1)+\sum_{b}Q(t)_{v,b} s_{b}(t{-}1)
   & \text{if $i=v$ and $\mu_{v}$ is red}.
  \end{cases}
\end{equation}
With this notation, \eqref{eq:kv-k} reads as
\begin{equation}
 \label{eq:k-kv-rel}
   k_{t}^{\vee}= k_{t}-\sum_{i} B(t{-}1)_{v,i} \,s_{i}(t{-}1)
   = k_{t}+\sum_{i} B(t{-}1)_{i,v}\, s_{i}(t{-}1).
\end{equation}
One can usually solve the linear relations \eqref{eq:k-s-rel} for
$s$-variables in terms of $k$-variables. If this is the case, the
mutation loop is called \emph{non-degenerate} (see \cite{KT2014}). Then
using \eqref{eq:kv-s-rel} and \eqref{eq:W-def}, we can express all the
mutation weights $\{W(m_{t})\}$ as functions of
$\seq{k}=(k_{1},\dots,k_{T})$.

Hereafter we assume that the mutation loop $\gamma$ is
non-degenerate. Then the ($\N^{n}$-graded) \emph{partition $q$-series}
associated with $\gamma$ is defined as
\begin{equation}
  \label{eq:zz-def}
  \ZZ(\gamma):=\sum_{\seq{k}\in \N^{T}} \biggl(
  \prod_{t=1}^{T}W(m_{t})\biggr)
   y^{\sum_{t=1}^{T}k_{t}\alpha_{t}}.
\end{equation}
\begin{rem}
 For a fixed $\beta\in \N^{n}$, there is only a finite number of
 $\seq{k}=(k_{1},\dots,k_{T})\in \N^{T}$ satisfying
 $\beta=\sum_{t=1}^{T}k_{t} \alpha_{t}$. Therefore $\ZZ(\gamma)$ is
 well-defined as an element of $\Ah_{Q}$.  In our previous paper
 \cite{KT2014}, we had no $\N^{n}$-grading and thus needed the
 additional ``positive'' assumption on the quadratic form in the
 mutation weight to guarantee this finiteness.
\end{rem}

\section{Backtracking invariance of the partition $q$-series}
\label{sec:backtracking}

Two successive mutations at the same vertex is called
\emph{backtracking}. It is well known that under backtracking a quiver
comes back to its original form: $\mu_{v}(\mu_{v}(Q))=Q$. All the
$c$-vectors recover their original values since
$\mu_{v}(\mu_{v}(\widetilde{Q}))=\widetilde{Q}$.

In this section, we prove that the partition series are invariant under
insertion or deletion of backtracking. The original version
\cite{KT2014} of partition $q$-series lacks this property; this is one
reason why we adopt different rules (e.g. \eqref{eq:s-evolution}) for
different signs (red/green).

\begin{theorem}
 \label{prop:Z-backtrack-inv} The partition $q$-series is invariant
 under insertion or deletion of backtracking:
 \begin{equation}
 \label{eq:Z-backtrack-inv} 
  \ZZ((Q;\seq{m}_{1}v\,v\,\seq{m}_{2},\varphi))=
  \ZZ((Q;\seq{m}_{1}\seq{m}_{2},\varphi)).
 \end{equation}
\end{theorem}
\begin{figure}[tbh]
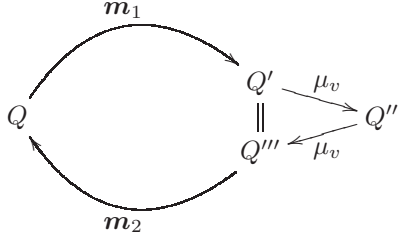

 \begin{equation*}
  \vcenter{ \xygraph{!{<0cm,0cm>;<23pt,0cm>:<0cm,13pt>::} 
   !{(0,0)}*+{Q}="Q0" 
   !{(4,1) }*+{Q'}="Q1" 
   !{(6,0)}*+{Q''}="Q2" 
   !{(4,-1) }*+{Q'''}="Q3" 
  "Q0":@/^1cm/"Q1"^{\textstyle {\seq{m}_{1}}}
  "Q3":@/^1cm/"Q0"^{\textstyle {\seq{m}_{2}}}
  "Q1":"Q2"^{\textstyle \mu_{v}} 
  "Q2":"Q3"^{\textstyle \mu_{v}} 
  "Q1":@{=}"Q3"}}
 \end{equation*}
 \caption{Mutation loop with backtracking.}
 \label{fig:backtracking}
\end{figure}
\begin{figure}[tbh]
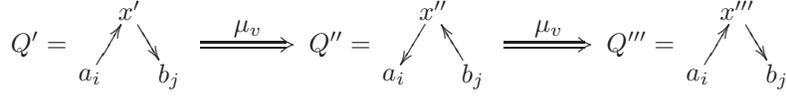

\begin{equation*}
 \xygraph{
!{<0cm,0cm>;<40mm,0cm>:<0cm,18mm>::}
!{(0,0)}*++{Q'=\vcenter{
\xybox{\xygraph{!{<0cm,0cm>;<15pt,0cm>:<0cm,25pt>::}
!{(1,1) }*+<3pt>{x'}="x" 
!{(0,0) }*+<3pt>{a_{i}}="a" 
!{(2,0) }*+<3pt>{b_{j}}="b" 
"a":"x" "x":"b" }}}}="Q0"
!{(1,0)}*++{Q''=\vcenter{\xybox{\xygraph{!{<0cm,0cm>;<15pt,0cm>:<0cm,25pt>::}
!{(1,1) }*+<3pt>{x''}="x" 
!{(0,0) }*+<3pt>{a_{i}}="a" 
!{(2,0) }*+<3pt>{b_{j}}="b" 
"x":"a" "b":"x" }}}}="Q1"
!{(2,0)}*++{Q'''=\vcenter{\xybox{\xygraph{!{<0cm,0cm>;<15pt,0cm>:<0cm,25pt>::}
!{(1,1) }*+<3pt>{x'''}="x" 
!{(0,0) }*+<3pt>{a_{i}}="a" 
!{(2,0) }*+<3pt>{b_{j}}="b" 
"a":"x" "x":"b" }}}}="Q2"
"Q0":@{=>}^{\textstyle \mu_{v}}"Q1" 
"Q1":@{=>}^{\textstyle \mu_{v}}"Q2" 
}\qquad ~
\end{equation*}
  \caption{Backtracking. Only the arrows incident on $v$ are
  shown. }\label{fig:mux^2}
\end{figure}
\begin{proof}
 The mutation loop $(Q;\seq{m}_{1}v\,v\,\seq{m}_{2},\varphi)$ is shown
 in Figure \ref{fig:backtracking}. We concentrate on two successive
 mutations constituting the backtracking:
 \begin{equation}
  \cdots \to Q' 
   \stackrel{\displaystyle\mu_{v}}{\longrightarrow} Q''
   \stackrel{\displaystyle\mu_{v}}{\longrightarrow} Q''' \to \cdots
 \end{equation}
 The proof is given only for the case when the signs of these two
 mutations are $(+,-) =$ (green, red); the other case $(-,+)$ is left to
 the reader. By assumption, the $c$-vector $\alpha$ of the vertex $v$
 changes as
 \begin{equation*}
  \alpha \mapsto -\alpha \mapsto  \alpha 
 \end{equation*} 
 for some $\alpha\in\N^{n}\setminus \{0\}$.

 Let $x'$, $x''$, $x'''$ be the $s$-variables associated with the vertex
 $v$ of $Q'$, $Q''$, $Q'''$ respectively, and $k_{1}$, $k_{2}$ be the
 $k$-variables corresponding to the two mutations at $v$.  As in Figure
 \ref{fig:mux^2}, we collectively denote by $i\to v$, $v\to j$ the
 arrows of $Q'$ touching $v$, and $a_{i}$, $b_{j}$ be the corresponding
 $s$-variables; some of the vertices $i$, $j$ may be missing, duplicated
 or identified. By \eqref{eq:k-s-rel}, \eqref{eq:kv-s-rel}, these
 $s$-variables are related with $k$- and $k^{\vee}$-variables as
 \begin{align*}
   k_{1}&=x'+x''-\sum a_{i},&\qquad& k^{\vee}_{1}=x'+x''-\sum b_{j}, \\
  k_{2}&=\sum a_{i}-(x''+x'''),&\qquad& k^{\vee}_{2}=\sum b_{j}
  -(x''+x'''). \\
 \end{align*}
 The weight corresponding to the backtracking $(v\,v)$ is given by
  \begin{align*}
    W((v\,v))=& \left( \frac{q^{\frac{1}{2}(x'+x''-\sum a_{i})(x'+x''-\sum
    b_{j})}}{(q)_{x'+x''-\sum a_{i}}} \right) \times \left(
    \frac{q^{-\frac{1}{2}(\sum a_{i} -(x''+x'''))(\sum b_{j}
    -(x''+x'''))}}{(q^{-1})_{\sum a_{i} -(x''+x''')}} \right) 
   \\ 
   & =
    \frac{q^{\frac{1}{2}k_{1}(k_{1}+\sum a_{i}-\sum
    b_{j})}}{(q)_{k_{1}}} \frac{q^{-\frac{1}{2}k_{2}(k_{2}+\sum
    b_{j}-\sum a_{i})}}{(q^{-1})_{k_{2}}} \\ & =
    \frac{q^{\frac{1}{2}(k_{1}^{2}-k_{2}^{2})}}{(q)_{k_{1}}(q^{-1})_{k_{2}}}
    \left( q^{\frac{1}{2}(\sum a_{i}-\sum b_{j})} \right)^{k_{1}+k_{2}}.
 \end{align*}
 By summing over $k_{1},k_{2}$ with appropriate $\N^{n}$-grading (see
 \eqref{eq:zz-def}), we have
 \begin{align*}
  & \sum_{k_{1},k_{2}=0}^{\infty}
    \frac{q^{\frac{1}{2}(k_{1}^{2}-k_{2}^{2})}}{(q)_{k_{1}}(q^{-1})_{k_{2}}}
    \left(q^{\frac{1}{2}(\sum a_{i}-\sum b_{j})}\right)^{k_{1}+k_{2}}
    y^{(k_{1}+k_{2})\alpha} \\ & =\sum_{n=0}^{\infty}
    \underbrace{\biggl( \sum_{\substack{k_{1},k_{2}\geq
    0\\k_{1}+k_{2}=n}} \frac{q^{\frac{1}{2}k_{1}^{2}}}{(q)_{k_{1}}}
    \frac{q^{-\frac{1}{2}k_{2}^{2}}}
    {(q^{-1})_{k_{2}}}\biggr)}_{\displaystyle=\delta_{n,0}}
    \left(q^{\frac{1}{2}(\sum a_{i}-\sum b_{j})}\right)^{n} y^{n\alpha}
    \\ & =1,
 \end{align*}
 where the identity \eqref{eq:EEinv-cor} of Corollary
 \ref{prop:EEinv-cor} is used.  The nontrivial contribution survives
 only when $k_{1}=k_{2}=0$, or equivalently $x'''=x'$ holds: all the
 $s$-variables on $Q'$ and $Q'''$ are vertex-wise equal.  Consequently,
 to evaluate $\ZZ((Q;\seq{m}_{1}v\,v\,\seq{m}_{2},\varphi))$, we can
 safely ignore the backtracking without changing its value; thus we have
 proved \eqref{eq:Z-backtrack-inv}.
\end{proof}

\section{Quantum dilogarithms and the combinatorial DT invariants}
\label{comb-DT}

Kontsevich and Soibelman \cite{Kontsevich2008,Kontsevich2010} developed
a general theory including motivic Donaldson-Thomas invariants,
wall-crossings, and cluster algebras. As a combinatorial analogue of
this, Keller \cite{Keller2011,Keller2012,Keller2013a} introduced and
studied reddening sequences and combinatorial Donaldson-Thomas
invariants.

In this section, we briefly summarize these notions and some known
facts.  The relationship with partition $q$-series is the main subject
of this paper and will be discussed in Section \ref{sec:Z-and-DT}.

\subsection{Quantum dilogarithms} 

A quantum dilogarithm series is defined by
\begin{equation}
 \EE(y;q) :=1+\frac{q^{1/2}}{q-1}y+\cdots
 +\frac{q^{n^{2}/2}}{(q^{n}-1)(q^{n}-q^{2})
 \cdots  (q^{n}-q^{n-1})}y^{n}+\cdots.
\end{equation}
It is also expressed as
\begin{align}
 \EE(y;q)
 & = \sum_{n=0}^{\infty} \frac{(-1)^{n}q^{n/2}}{(q)_{n}}y^{n}
 = \sum_{n=0}^{\infty} \frac{q^{-\frac{1}{2}n^{2}}}{(q^{-1})_{n}}y^{n}
 =\prod_{n=0}^{\infty} \frac{1}{1+q^{n+\frac{1}{2}}y}
 \\ & 
 =\exp\left(
   \sum_{k=1}^{\infty} \frac{(-y)^k  }{k (q^{-k/2} - q^{k/2})}
   \right).
 \label{eq:E=exp()}
\end{align}
We will mostly use the following form
\begin{equation}
  \EE(y;q^{-1})
 = \sum_{n=0}^{\infty} \frac{q^{\frac{1}{2}n^{2}}}{(q)_{n}}y^{n},\qquad 
  \EE(y;q)
 = \sum_{n=0}^{\infty} \frac{q^{-\frac{1}{2}n^{2}}}{(q^{-1})_{n}}y^{n}
\end{equation}

It is well-known that if $u$ and $v$ satisfies the relation $uv=qvu$,
then the pentagon identity holds
\cite{Schuetzenberger1953,Faddeev1993,Faddeev1994}:
\begin{equation}
 \label{eq:EE=EEE}
 \EE(u;q)\EE(v;q)=
  \EE(v;q)\EE(q^{-1/2}uv;q)\EE(u;q).
\end{equation}

\subsection{Reddening sequences} 
\label{sec:reddening-DT}

Consider a mutation sequence $\seq{m}=(m_{1},m_{2},\dots,m_{T})$ on a quiver $Q$:
\begin{equation}
 \label{eq:Q-seq-3}
 \xygraph{!{<0cm,0cm>;<18mm,0cm>:<0cm,20pt>::}
!{(0,0)}*+{Q(0)}="q0"
!{(1,0)}*+{Q(1)}="q1"
!{(2,0)}*+{~\cdots~}="q2"
!{(2.9,0)}*+{Q(t{-}1)}="q3"
!{(4.1,0) }*+{Q(t)}="q4"
!{(5,0)}*+{~\cdots~}="q5"
!{(6,0) }*+{Q(T)}="q6"
"q0":^{\textstyle \mu_{m_{1}}}"q1"
"q1":^{\textstyle \mu_{m_{2}}}"q2"  
"q2":"q3"  
"q3":^{\textstyle \mu_{m_{t}}}"q4"  
"q4":"q5"  
"q5":^{\textstyle \mu_{m_{T}}}"q6"  
}.
\end{equation}
The mutation sequence $\seq{m}$ is called \emph{reddening} if all
vertices of the final quiver $Q(T)$ are red.  Clearly, all maximal green
sequences are reddening, but the latter class is much wider. Not all
quivers admit reddening sequences.

The following facts are known:
\begin{theorem}[Keller\cite{Keller2012}]
 \label{thm:EEm-reddening} 
 If $\seq{m}$ and $\seq{m}'$ are reddening
 sequences on the quiver $Q$, then there is a frozen isomorphism between
 the final ice quivers $\mu_{\seq{m}}(Q^{\wedge})\simeq
 \mu_{\seq{m}'}(Q^{\wedge})$.
\end{theorem}
\begin{theorem}[Br{\"u}stle--Dupont--P{\'e}rotin \cite{Bruestle2013}]
 \label{thm:reddening=rev-end} Let $\seq{m}=(m_{1},\cdots,m_{T})$ be a
 reddening sequence. Then the associated final ice quiver
 $\widetilde{Q}(T)=\mu_{\seq{m}} (Q^{\wedge})$ is frozen isomorphic to a
 co-framed quiver $Q^{\vee}$, that is, there is a permutation $\varphi$
 of $Q_{0}=\{1,2,\dots,n\}$ such that
 \begin{enumerate}[(i)]
  \item $\varphi$ represents an isomorphism of quivers $Q(T)\simeq
	Q(0)$, and
  \item in the final ice quiver $\widetilde{Q}(T)$, $i'\to \varphi(i)$
	is the only arrow starting from $i'$ and there is no arrow
	pointing to $i'$. In terms of $c$-vectors, we have
	$c_{i}(T)=-e_{\varphi(i)}$.
 \end{enumerate}
\end{theorem}
Thanks to Theorem \ref{thm:reddening=rev-end}, every reddening sequence
$\seq{m}$ is naturally associated with a mutation loop: we can use
$\varphi$ as a boundary condition to make up a mutation loop $\gamma
=(Q;\seq{m},\varphi)$, which will be called \emph{reddening mutation
loop} corresponding to $\seq{m}$.

\subsection{Combinatorial DT invariants}

Following Keller, let us associate a quantum dilogarithm $
\EE(y^{\alpha_{t}};q^{\varepsilon_{t}})$ for each mutation
$\mu_{m_{t}}:Q(t-1)\to Q(t)$. Here $\varepsilon_{t}$ is the sign of
$\mu_{m_{t}}$ (see \eqref{eq:eps-def}) and
$\alpha_{t}=\varepsilon_{t}c_{m_{t}}(t-1)\in \N^{n}$ is the
(sign-corrected) $c$-vector (see \eqref{eq:alpha-def}).  For a mutation
sequence \eqref{eq:Q-seq-3}, we consider the following ordered product
of these:
\begin{equation}
  \EE(Q;\seq{m}):=
  \EE(y^{\alpha_{1}};q^{\varepsilon_{1}})
  \EE(y^{\alpha_{2}};q^{\varepsilon_{2}})
  \cdots 
  \EE(y^{\alpha_{T}};q^{\varepsilon_{T}})\quad \in \Ah_{Q}.
\end{equation}
Here, the algebra $\Ah_{Q}$ and the skew-symmetric form $\vev{~,~}$ are
always defined in terms of the initial quiver $Q=Q(0)$.  The product of
quantum dilogarithms enjoy the following remarkable property.
\begin{theorem}[Keller\cite{Keller2012}, Nagao\cite{Nagao2011}]
 \label{thm:EEm=EEm'} If $\seq{m}$ and $\seq{m}'$ are two mutation
 sequences such that there is a frozen isomorphism between
 $\mu_{\seq{m}}(Q^{\wedge})$ and $\mu_{\seq{m}'}(Q^{\wedge})$, then we
 have $\EE(Q;\seq{m})=\EE(Q;\seq{m}')$.
\end{theorem}
Theorems \ref{thm:EEm-reddening} and \ref{thm:EEm=EEm'} imply that if
$Q$ admits a reddening sequence $\seq{m}$, the power series
\begin{equation}
  \EE_{Q}:= \EE(Q;\seq{m})\in \Ah_{Q}.
\end{equation}
is independent of the choice of the reddening sequence $\seq{m}$ and is
canonically associated with $Q$.  Keller \cite{Keller2013a} named this
invariant as \emph{combinatorial Donaldson-Thomas (DT) invariant}.


The pentagon identity \eqref{eq:EE=EEE} is nothing but the combinatorial
DT invariant of $A_{2}$ quiver $Q=(1{\to}2)$ corresponding to the two
reddening sequences $\seq{m}=(1,2)$ and $\seq{m}'=(2,1,2)$ depicted in
Figure \ref{fig:pentagon}.  $\gamma=(Q;(1,2), \id)$ and
$\gamma'=(Q;(2,1,2), (12))$ are the reddening mutation loops
corresponding to $\seq{m}$ and $\seq{m}'$, respectively.

\begin{rem}
 The statements of Theorems \ref{thm:EEm-reddening} and
 \ref{thm:reddening=rev-end} are combinatorial, but the known proofs are
 based on categorification in terms of Ginzburg dg-algebra
 \cite{Ginzburg2006}.
\end{rem}

\section{Partition $q$-series and the combinatorial DT invariant}
\label{sec:Z-and-DT}

We have seen that every reddening sequence $\seq{m}$ is canonically
associated with a mutation loop $\gamma=(Q;\seq{m},\varphi)$.  It is
thus natural to compare the combinatorial $DT$ invariant
$\EE_{Q}=\EE(Q;\seq{m})$ with the partition $q$-series $\ZZ(\gamma)$.

In this section, we show that there is a precise match between the
partition $q$-series and the product of quantum dilogarithms for any
reddening mutation sequence.  Therefore, the partition $q$-series
provide a state-sum interpretation (fermionic sum formula) for the
combinatorial Donaldson-Thomas invariants.

The following theorem is the main result of this paper.
\begin{theorem}
 \label{thm:main} Let $\gamma=(Q;\seq{m},\varphi)$ be a reddening
 mutation loop. Then, the partition $q$-series and the
 combinatorial Donaldson-Thomas invariant are related as
 \begin{equation}
  \label{eq:ee=zz-rev-end}
   \boxed{\quad
   \rule[-7pt]{0pt}{20pt}
   \ZZ(\gamma)=\overline{\EE(Q;\seq{m})}.
   \quad
   }
 \end{equation}
 Here $\overline{\phantom{x}}:\Ah_{Q}\to \Ah_{Q}$ is a $\Q$-algebra
 anti-automorphism defined in \eqref{eq:anti-autom}.
\end{theorem}

Section \ref{sec:Examples} contains various examples of partition
$q$-series covered by Theorem \ref{thm:main}.  The rest of this section
is devoted to the proof of Theorem \ref{thm:main}.

\begin{rem}
In the paper \cite{Cecotti2010}, Cecotti-Neitzke-Vafa propose
a relation between four-dimensional gauge theories and parafermionic
conformal field theories. In particular, they found a wonderful
observation that the canonical trace of a special product of quantum
dilogarithms associated with a Dynkin diagram is written in terms of
characters. It would be interesting to find a precise relation with
their work.
\end{rem}

\subsection{Evolution along mutation sequence}

In this subsection, we collect some results how the $s$-variables
$\{s_{i}(t)\}$ and $c$-vectors $\{c_{i}(t)\}$ evolve along the mutation
sequence \eqref{eq:Q-seq-3}. These are needed to keep track of $\N^{n}$
grading of partition $q$-series. 

Proposition \ref{prop:CBC=B} is due to Nakanishi-Zelevinsky
\cite{Nakanishi2012a}. A proof is given here to make this paper
self-contained.
\begin{prop}
 \label{prop:CBC=B} For the mutation sequence \eqref{eq:Q-seq-3}, we have  
 \begin{equation}
  \label{eq:Bcc-t}
  B(t)_{ij}=\vev{c_{i}(t),c_{j}(t)},\qquad (0\leq t\leq T)
 \end{equation}
 or equivalently, 
 \begin{equation}
  \label{eq:CBC=B}
   C(t)B(0)C(t)^{\top} =B(t).\qquad  (0\leq t\leq T)
 \end{equation}
\end{prop}
\begin{proof}
 We prove \eqref{eq:Bcc-t} by induction on $t$. The case $t=0$ is clear
 from \eqref{def:vev} and $c_{i}(0)=e_{i}$. Assuming \eqref{eq:Bcc-t} to
 hold for $t$, we consider the mutation $\mu_{v}:Q(t)\to Q(t+1)$ where
 $v=m_{t+1}$. By skewness of $\vev{~,~}$, it suffices to consider
 the following four cases:
 \begin{itemize}
  \itemsep=2ex 
  \item (case A-1) $i=v$, $j\neq v$.
	 \begin{align*}
	  & \vev{c_{v}(t+1),c_{j}(t+1)} 
	  = \vev{
	  -c_{v}(t),c_{j}(t)+ Q(t)_{j,v} c_{v}(t)}
	  & \quad\text{(using \eqref{eq:c-vec-change})}
	  \\
	  & = - \vev{c_{v}(t),c_{j}(t)}	  
	  & \llap{(\text{skew-symmetry of $\vev{~,~}$})}
	  \\	  & 
	  = -B(t)_{v,j}	
	  & \llap{(\text{by induction hypothesis})}
	  \\
	  & =  B(t+1)_{v,j}.
	  & \llap{(\text{using \eqref{eq:matrix-mutation}})}
	 \end{align*}

  \item (case A-2) $i\neq v$, $j=v$. The proof closely parallels that of
       (case A-1).

  \item (case B-1) $i\neq v$, $j\neq v$, $\mu_{v}$ is green mutation
	 \begin{align*}
	  & \vev{c_{i}(t+1),c_{j}(t+1)}
	  \\ &
	  = \vev{
	  c_{i}(t)+Q(t)_{i,v}c_{v}(t),
	  c_{j}(t)+Q(t)_{j,v}c_{v}(t)
	  }
	  &
	  \llap{\text{(using \eqref{eq:c-vec-change})}}
	  \\[.5ex]
	  & = \vev{c_{i}(t),c_{j}(t)} + 
	  Q(t)_{i,v} \vev{c_{v}(t), c_{j}(t)}
	  +
	  Q(t)_{j,v} \vev{ c_{i}(t), c_{v}(t) }
	  \\
	  &
	  & \llap{(\text{skew-symmetry of $\vev{~,~}$})}
	  \\[.5ex]
	  & = B(t)_{i,j} +
	  Q(t)_{i,v} B(t)_{v,j}
	  +
	  Q(t)_{j,v} B(t)_{i,v}
	  &  \quad \llap{(\text{by induction hypothesis})}
	  \\[.5ex]
	  & = B(t)_{i,j} +
	  Q(t)_{i,v} (Q(t)_{v,j}-Q(t)_{j,v})
	  +
	  Q(t)_{j,v} (Q(t)_{i,v}-Q(t)_{v,i})
	  \\[.5ex]
	  & = B(t)_{i,j} +
	  Q(t)_{i,v} Q(t)_{v,j}
	  - Q(t)_{j,v}Q(t)_{v,i}
	  \\[.5ex]
	  & =  B(t+1)_{i,j}.
	  & \llap{(\text{using \eqref{eq:matrix-mutation}})}
	 \end{align*}
  \item (case B-2)  $i\neq v$, $j\neq v$, $\mu_{v}$ is red
	mutation. The proof is similar to that of (case B-1).
 \end{itemize}
 In conclusion, \eqref{eq:Bcc-t} is also true for $t+1$.
\end{proof}

Since $c$-vectors $\{c_{i}(t)\}$ form a basis of $\Z^{n}$ for each $t$,
it is natural to introduce the \emph{state vector} of $Q(t)$ defined by
\begin{equation}
 \label{eq:psi-def} \psi(t):=\sum_{i=1}^{n} s_{i}(t) c_{i}(t) ~\in
 \Z^{n}\qquad (0\leq t\leq T).
\end{equation}
\begin{prop} 
 \label{prop:psi-change} 
 Along the mutation sequence \eqref{eq:Q-seq-3},
 the state vector changes as
 \begin{equation}
  \label{eq:psi-change} 
   \psi(t)=
   \psi(t{-}1)   -k_{t} \alpha_{t},\qquad (t=1,\dots,T).
 \end{equation}
\end{prop}
\begin{proof} 
 There are two cases to be considered.

 Case 1) $\mu_{m_{t}}:Q(t-1)\to Q(t)$ is green ($\varepsilon_{t}=+1$):
  \begin{align*}
   \psi(t)&=\sum_{i} s_{i}(t)c_{i}(t)
   \\[-1ex]
   & =s_{m_{t}}(t)c_{m_{t}}(t)
   +\sum_{i\neq m_{t}} s_{i}(t)c_{i}(t)
   \\[-1ex]
   & = \Bigl(k_{t}-s_{m_{t}}(t{-}1)+\sum_{a} Q(t{-}1)_{a,m_{t}} 
   s_{a}(t{-}1)\Bigr)(-c_{m_{t}}(t{-}1))
   &&\qquad& (\text{by \eqref{eq:s-evolution}})
   \\[-1ex]
   & \qquad    +\sum_{i\neq m_{t}} s_{i}(t{-}1) 
   \bigl(c_{i}(t{-}1)+Q(t{-}1)_{i,m_{t}}c_{m_{t}}(t{-}1)\bigr) 
   &&& (\text{by \eqref{eq:c-vec-change}})
   \\[-1ex]
   & = \bigl(k_{t}-s_{m_{t}}(t{-}1)\bigr)(-c_{m_{t}}(t{-}1))
   +\sum_{i\neq m_{t}} s_{i}(t{-}1)
   \bigl(c_{i}(t{-}1)\bigr)
   \\[-1ex]
   & = -k_{t} c_{m_{t}}(t{-}1)+\sum_{i} s_{i}(t{-}1)c_{i}(t{-}1)
   \\[-.5ex]
   & =\psi(t{-}1) -k_{t} \varepsilon_{t}c_{m_{t}}(t{-}1)
   \\[.5ex]
   & = \psi(t{-}1)-k_{t} \alpha_{t}.
  \end{align*}

 Case 2) $\mu_{m_{t}}:Q(t-1)\to Q(t)$ is red ($\varepsilon_{t}=-1$):
  \begin{align*}
   \psi(t)&=\sum_{i} s_{i}(t)c_{i}(t)
   \\[-1ex]
   & =s_{m_{t}}(t)c_{m_{t}}(t)
   +\sum_{i\neq m_{t}} s_{i}(t)c_{i}(t)
   \\[-1ex]
   & = \Bigl(-k_{t}-s_{m_{t}}(t{-}1)+\sum_{b} Q(t{-}1)_{m_{t},b}
   s_{b}(t{-}1)\Bigr)(-c_{m_{t}}(t{-}1))
   &&\quad& (\text{by \eqref{eq:s-evolution}})
   \\[-1ex]
   & \qquad    
 +\sum_{i\neq m_{t}} s_{i}(t{-}1) 
   \bigl(c_{i}(t{-}1)+Q(t{-}1)_{m_{t},i}c_{m_{t}}(t{-}1)\bigr) 
   &&& (\text{by \eqref{eq:c-vec-change}})
   \\[-1ex]
   & = \bigl(-k_{t}-s_{m_{t}}(t{-}1)\bigr)(-c_{m_{t}}(t{-}1))
   +\sum_{i\neq m_{t}} s_{i}(t{-}1)
   \bigl(c_{i}(t{-}1)\bigr)
   \\[-1ex]
   & = +k_{t} c_{m_{t}}(t{-}1)+\sum_{i} s_{i}(t{-}1)c_{i}(t{-}1)
   \\[-.5ex]
   & =\psi(t{-}1) -k_{t} \varepsilon_{t}c_{m_{t}}(t{-}1)
   \\[.5ex]
   & = \psi(t{-}1)-k_{t} \alpha_{t}.
  \end{align*}
\end{proof}

Therefore $\N^{n}$-grading of the partition $q$-series expresses the
total change of the state vector around the mutation loop:
\begin{cor}
 \label{prop:psi-in-out} The state vectors of the initial and the final
 quivers are related as
 \begin{equation*}
    \label{eq:psi-in-out}
     \psi(0)-\psi(T)=\sum_{t=1}^{T}k_{t}\alpha_{t}.
 \end{equation*}
\end{cor}

\begin{lemma}
 \label{prop:rev-end} Let $\gamma=(Q;\seq{m},\varphi)$ be a reddening
 mutation loop. Then,
 \begin{enumerate}[(i)]
  \item The state vectors $\{\psi(t)\}_{t=0}^{T}$ are anti-periodic
	along the loop, that is,
	\begin{equation*}
	 \label{eq:psi-anti-periodic}
	 \psi(T)=-\psi(0).
	\end{equation*}

 \item The mutation loop $\gamma$ is non-degenerate. In particular, the
       initial $s$-variables are
       expressed as
       \begin{equation*}
	 \seq{s}(0):=(s_{1}(0),\dots,s_{n}(0))
	  =\frac{1}{2}\sum_{t=1}^{T}k_{t}\alpha_{t}.
       \end{equation*}
 \end{enumerate}
\end{lemma}

\begin{proof}
 By the boundary condition $\varphi$, the initial and the final
 $s$-variables are identified as $s_{i}(T)=s_{\varphi(i)}(0)$.  By (ii)
 of Theorem \ref{thm:reddening=rev-end}, we have also
 $c_{i}(T)=-c_{\varphi(i)}(0)$. Therefore
 \begin{equation*}
  \psi(T)=\sum_{i=1}^{n} s_{i}(T)c_{i}(T) =- \sum_{i=1}^{n}
   s_{\varphi(i)}(0)c_{\varphi(i)}(0) =- \sum_{j=1}^{n}
   s_{j}(0)c_{j}(0)= -\psi(0).
 \end{equation*}
 This proves (i). We have then
 \begin{align*}
  \seq{s}(0)&=\sum_{i=1}^{n} s_{i}(0)e_{i}
  =\sum_{i=1}^{n} s_{i}(0)c_{i}(0) =\psi(0)
  =\frac{1}{2}(\psi(0)-\psi(T)) 
  =\frac{1}{2}\sum_{t=1}^{T}k_{t}\alpha_{t},
 \end{align*}
 where the last equality is by Corollary \ref{prop:psi-in-out}. Thus the
 initial $s$-variables are expressed in terms of $k$-variables
 alone. We can obtain similar formulas for the remaining 
 $s$-variables by recursive use of the relation \eqref{eq:s-evolution}. This
 proves (ii).
\end{proof}

The following relation will play a key role in the proof of Theorem
\ref{thm:main}.
\begin{prop}
 \label{prop:kkv-k2-rel} For any mutation sequence, we have
 \begin{equation}
  \label{eq:kkv-k2-rel}
   \rule[-6pt]{0pt}{20pt}
   \sum_{t=1}^{T}\varepsilon_{t}k_{t}k_{t}^{\vee}
   +\vev{\psi(0), \psi(T)} 
   =  \sum_{t=1}^{T}\varepsilon_{t}k_{t}^{2}  
   - \sum_{1\leq i< j\leq
   T}k_{i}k_{j}\vev{\alpha_{i},\alpha_{j}}.
 \end{equation}
\end{prop}
\begin{proof}
 \begin{align*}
   & \sum_{t=1}^{T}\varepsilon_{t}k_{t}k_{t}^{\vee}-
   \sum_{t=1}^{T}\varepsilon_{t}k_{t}^{2}  
   =  \sum_{t=1}^{T}\varepsilon_{t}k_{t}
   (k_{t}^{\vee}-k_{t})
   \\
   &=  \sum_{t=1}^{T}\varepsilon_{t}k_{t}
   \sum_{i=1}^{n} B(t{-}1)_{i,m_{t}} s_{i}(t{-}1)
   & (\text{by \eqref{eq:k-kv-rel}})
   \\
   &=  \sum_{t=1}^{T}\varepsilon_{t}k_{t}
   \sum_{i=1}^{n} \vev{c_{i}(t{-}1),c_{m_{t}}(t{-}1)} s_{i}(t{-}1)
   & (\text{by \eqref{eq:Bcc-t}})
   \\
   &=  \sum_{t=1}^{T}k_{t}
   \sum_{i=1}^{n} \vev{c_{i}(t{-}1),\alpha_{t}} s_{i}(t{-}1)
      & (\text{by \eqref{eq:alpha-def}})
   \\
   &=  \sum_{t=1}^{T}k_{t}
    \vev{ \psi(t{-}1),\alpha_{t}}
   & (\text{by \eqref{eq:psi-def}})
   \\
   &=  \sum_{t=1}^{T}k_{t}
    \vev{ \psi(0)-\sum_{i=1}^{t-1}k_{i}\alpha_{i},\alpha_{t}}
   & (\text{by \eqref{eq:psi-change}})
   \\
   &=  \sum_{t=1}^{T}k_{t}
    \vev{ \psi(0),\alpha_{t}}
   - \sum_{t=1}^{T}\sum_{i=1}^{t-1} k_{i} k_{t}
   \vev{ \alpha_{i},\alpha_{t}}
   \\
   &=  
    \bvev{ \psi(0),\sum_{t=1}^{T}k_{t}\alpha_{t}}
   - \sum_{j=1}^{T}\sum_{i=1}^{j-1} k_{i} k_{j}
   \vev{ \alpha_{i},\alpha_{j}}
   \\
   &=  
    \vev{ \psi(0),\psi(0)-\psi(T)}
   - \sum_{1\leq i< j\leq T}k_{i} k_{j}
   \vev{ \alpha_{i},\alpha_{j}}
   & (\text{by Corollary \ref{prop:psi-in-out}})
   \\
   &=  
   - \vev{ \psi(0),\psi(T)}
   - \sum_{1\leq i< j\leq T}k_{i} k_{j}
   \vev{ \alpha_{i},\alpha_{j}}.
   & (\text{by skewness of $\vev{~,~}$})
 \end{align*}
 By arranging the terms, we obtain \eqref{eq:kkv-k2-rel}.
\end{proof}

\subsection{Proof of Theorem \ref{thm:main}}

We are now ready to prove Theorem \ref{thm:main}.  
The partition $q$-series associated with the loop $\gamma$ is defined to be
\begin{equation}
 \label{eq:zz-reverse}
  \begin{split}
  \ZZ(\gamma) &=   \sum_{\seq{k}\in \N^{T}} \prod_{t=1}^{T}W(m_{t})  \;
   y^{\sum_{t=1}^{T}k_{t}\alpha_{t}}
   =
   \sum_{\seq{k}\in \N^{T}} 
  \frac{q^{\frac{1}{2}\sum_{t=1}^{T}\varepsilon_{t}k_{t}k_{t}^{\vee}}}
   {\prod_{t=1}^{T} (q^{\varepsilon_{t}})_{k_{t}}}\,
   y^{\sum_{t=1}^{T}k_{t}\alpha_{t}}.
  \end{split}
\end{equation}
On the other hand, the quantum dilogarithm product along $\seq{m}$ is
given by
\begin{equation*}
 \begin{split}
  \EE(Q;\seq{m}) &=
  \EE(m_{1};q^{\varepsilon_{1}})\EE(m_{2};q^{\varepsilon_{2}})
  \cdots \EE(m_{T};q^{\varepsilon_{T}})
  \\
  &=
  \biggl(
  \sum_{k_{1}=0}^{\infty}
  \frac{q^{-\frac{\varepsilon_{1}}{2}k_{1}^{2}}}
  {(q^{-\varepsilon_{1}})_{k_{1}}}
  y^{k_{1}\alpha_{1}}\biggr)
  \cdots
  \biggl(
  \sum_{k_{T}=0}^{\infty}
  \frac{q^{-\frac{\varepsilon_{T}}{2}k_{T}^{2}}}
  {(q^{-\varepsilon_{T}})_{k_{T}}}
  y^{k_{T}\alpha_{T}}\biggr)
  \\
  &
  =
  \sum_{\seq{k}\in \N^{T}}
  \frac{q^{-\frac{1}{2}\sum_{t=1}^{T}\varepsilon_{t}k_{t}^{2}}}
  {\prod_{t=1}^{T} (q^{-\varepsilon_{t}})_{k_{t}}}\,
  y^{k_{1}\alpha_{1}}
  \cdots
  y^{k_{T}\alpha_{T}}
  \\
  &
  =
  \sum_{\seq{k}\in \N^{T}}
  \frac{q^{-\frac{1}{2}\sum_{t=1}^{T}\varepsilon_{t}k_{t}^{2}}}
  {\prod_{t=1}^{T} (q^{-\varepsilon_{t}})_{k_{t}}}
  q^{\frac{1}{2}\sum_{1\leq i< j\leq T}k_{i}k_{j}
  \vev{\alpha_{i},\alpha_{j}}}\,
  y^{\sum_{t=1}^{T}k_{t}\alpha_{t}}.
 \end{split}
\end{equation*}
Therefore,
\begin{equation}
 \label{eq:E-bar-rev} \overline{\strut \EE(Q;\seq{m})} =
 \sum_{\seq{k}\in \N^{T}}
 \frac{q^{\frac{1}{2}\sum_{t=1}^{T}\varepsilon_{t}k_{t}^{2}}}{\prod_{t=1}^{T}
 (q^{\varepsilon_{t}})_{k_{t}}} q^{-\frac{1}{2} \sum_{1\leq i< j\leq
 T}k_{i}k_{j}\vev{\alpha_{i},\alpha_{j}}}\,
 y^{\sum_{t=1}^{T}k_{t}\alpha_{t}}.
\end{equation}
 Thus, all we have to show is that the exponents of $q$ in the summands
 of \eqref{eq:zz-reverse} and \eqref{eq:E-bar-rev} are equal for every
 $\seq{k}$:
 \begin{equation}
  \label{eq:zz-ee}
  \sum_{t=1}^{T}\varepsilon_{t}k_{t}k_{t}^{\vee}
   =  \sum_{t=1}^{T}\varepsilon_{t}k_{t}^{2}  
   - \sum_{1\leq i< j\leq T}k_{i}k_{j}\vev{\alpha_{i},\alpha_{j}}.
\end{equation}
Indeed, by Lemma \ref{prop:rev-end} (i), we have $\psi(T)=-\psi(0)$,
which implies $ \vev{\psi(0),\psi(T)}=-\vev{\psi(0),\psi(0)}=0$ by the
skewness of $\vev{~,~}$. Then \eqref{eq:zz-ee} follows from Proposition
\ref{prop:kkv-k2-rel}. This completes the proof of Theorem
\ref{thm:main}.

\section{Examples}
\label{sec:Examples}

In this section, we collect various examples of the reddening mutation
loops and the associated partition $q$-series to illustrate Theorem
\ref{thm:main}.

\subsection{$A_{2}^{(1)}$-quiver}

As a simplest example of quiver with an oriented cycle, let us take the
$A_{2}^{(1)}$ quiver
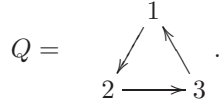
\begin{equation*}
 Q=\quad\vcenter{\xybox{ \xygraph{!{<0cm,0cm>;<7mm,0cm>:<0cm,7mm>::}
 !{a(90)}*+{1}="1" !{a(210) }*+{2}="2" !{a(330) }*+{3}="3" "1":"2"
 "2":"3" "3":"1" }}}.
\end{equation*}
By performing successive mutations on $\widetilde{Q}(0):=Q^{\vee}$
(Figure \ref{fig:A2(1)-loop}), it is easy to see that the mutation
sequence
\begin{equation*}
 \seq{m}=(1,2,3,1)
\end{equation*}
is maximal green, reddening sequence with the boundary condition 
\begin{equation}
 \label{eq:A2(1)-bc}
 (13)=(1\mapsto 3, ~2\mapsto 2, ~3\mapsto 1)\in S_{3}.
\end{equation}
From Figure \ref{fig:A2(1)-loop} we can read off  the $c$-vectors of the
mutating vertices:
\begin{align*}
\alpha _1  = c_{1}(0)=   (1,0,0), &&&
\alpha _2  = c_{2}(1)=   (0,1,0), \\
\alpha _3  = c_{3}(2)=   (1,0,1), &&&
\alpha _4  = c_{1}(3)=   (0,0,1). 
\end{align*}
The $s$-variables change as follows (cf. \eqref{eq:s-evolution}):
\begin{equation*}
 \renewcommand{\arraystretch}{1.1}
 \begin{array}{c|c|c|c}
  & 1 & 2 & 3 \\
  \hline
 Q(0) & s_1 & s_2 & s_3 \\
 Q(1) & s_{1}'=k_{1}-s_1+s_3 & s_2 & s_3 \\
 Q(2) & s_{1}'=k_1-s_1+s_3 & s_{2}'=k_2-s_2 & s_3 \\
 Q(3) & s_{1}'=k_1-s_1+s_3 & s_{2}'=k_2-s_2 & s_{3}'=k_3-s_{3}+s_{1}' \\
 Q(4) & s_{1}''=k_4-s_{1}'+s_{3}' & s_{2}'=k_2-s_2 
  & s_{3}'=k_3-s_{3}+s_{1}'
\end{array}
\end{equation*}
The boundary condition \eqref{eq:A2(1)-bc} imposes
\begin{equation*}
 s_{1}''=s_3,\qquad s_{2}'=s_2,\qquad
  s_{3}'=s_1.
\end{equation*}
From these relations, we can express $s$-variables in terms of
$k$-variables:
\begin{equation*}
 s_1 = s_{3}'= \frac{1}{2} \left(k_1+k_3\right),\quad
  s_{1}'=\frac{1}{2}(k_{1}+k_{4}),\quad 
 s_2 = s_{2}'= \frac{k_2}{2},\quad 
 s_3 = s_{1}''= \frac{1}{2} \left(k_3+k_4\right).
\end{equation*}
The $k^{\vee}$-variables are then
\begin{equation*}
\renewcommand{\arraystretch}{1.2}
\arraycolsep=2pt
\begin{array}{rclcl}
 k_{1}^{\vee} &=& s_{1}+s_{1}'-s_{2} &=&
 k_1-\frac{k_2}{2}+\frac{k_3}{2}+\frac{k_4}{2},
 \\
 k_{2}^{\vee} &=&  s_{2}+s_{2}'-s_{1}' &=&
 -\frac{k_1}{2}+k_2-\frac{k_4}{2},
 \\
 k_{3}^{\vee} &=&  s_{3}+s_{3}' &=&
 \frac{k_1}{2}+k_3+\frac{k_4}{2},
 \\
 k_{4}^{\vee} &=&  s_{1}'+s_{1}''-s_{2}' &=&
 \frac{k_1}{2}-\frac{k_2}{2}+\frac{k_3}{2}+k_4.
\end{array}
\end{equation*}
Plugging these into the definition of mutation weights \eqref{eq:W-def}
and summing over $k$-variables, we obtain
\begin{equation}
 \label{eq:Z-A2(1)}
 \ZZ(\gamma)=\sum_{\seq{k}\in \N^{4}}
  \frac{q^{\frac{1}{2}(k_1^2+k_2^2+k_3^2+k_4^2-k_1 k_2 + k_1 k_3
  + k_1 k_4 -k_2 k_4+k_3 k_4)}}
  {(q)_{k_{1}}(q)_{k_{2}}(q)_{k_{3}}(q)_{k_{4}}} y^{(k_1+k_3,k_2,k_3+k_4)}.
\end{equation}

\begin{figure}[t]
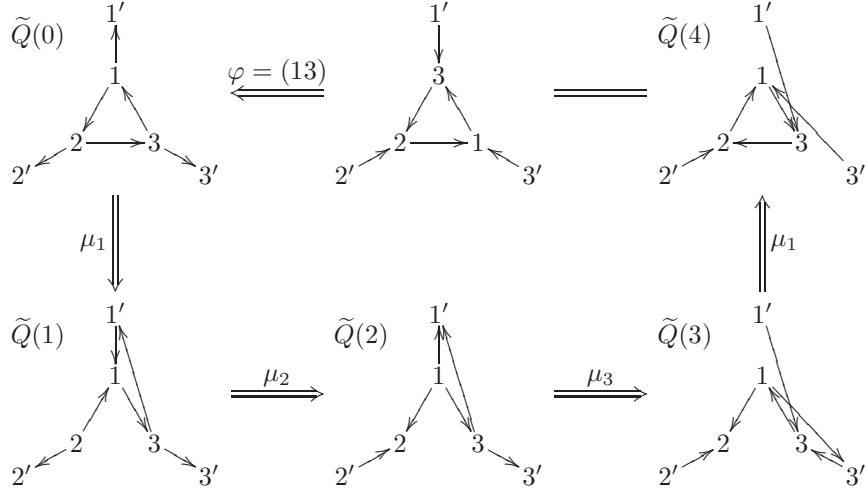

\begin{equation*}
 \xygraph{!{<0cm,0cm>;<43mm,0cm>:<0cm,40mm>::}
  !{(0,1)-(0.24,-0.2)}*+{\widetilde{Q}(0)}
  !{(0,1)}*+{ 
  \vcenter{\xybox{
  \xygraph{!{<0cm,0cm>;<6mm,0cm>:<0cm,6mm>::}
  !{a(90)}*+<2.8pt>{1}="1" 
  !{a(210) }*+<2.8pt>{2}="2" 
  !{a(330) }*+<2.8pt>{3}="3" 
  !{(0,0);a(90)**{}?(2.4)}*+<2.8pt>{1'}="1'" 
  !{(0,0);a(210)**{}?(2.4)}*+<2.8pt>{2'}="2'" 
  !{(0,0);a(330)**{}?(2.4)}*+<2.8pt>{3'}="3'" 
  "1":"2" 
  "2":"3" 
  "3":"1" 
  "1":"1'" 
  "2":"2'" 
  "3":"3'" 
  }}}}="q0" 
  !{(0,0)-(0.24,-0.2)}*+{\widetilde{Q}(1)}
  !{(0,0)}*+{
  \vcenter{\xybox{
 \xygraph{!{<0cm,0cm>;<6mm,0cm>:<0cm,6mm>::}
  !{a(90)}*+<2.8pt>{1}="1" 
  !{a(210) }*+<2.8pt>{2}="2" 
  !{a(330) }*+<2.8pt>{3}="3" 
  !{(0,0);a(90)**{}?(2.4)}*+<2.8pt>{1'}="1'" 
  !{(0,0);a(210)**{}?(2.4)}*+<2.8pt>{2'}="2'" 
  !{(0,0);a(330)**{}?(2.4)}*+<2.8pt>{3'}="3'" 
  "1":"3" 
  "2":"1" 
  "2":"2'" 
  "3":"1'" 
  "3":"3'" 
  "1'":"1" 
  }}}}="q1" 
  !{(1,0)-(0.24,-0.2)}*+{\widetilde{Q}(2)}
  !{(1,0)}*+{
  \vcenter{\xybox{
 \xygraph{!{<0cm,0cm>;<6mm,0cm>:<0cm,6mm>::}
  !{a(90)}*+<2.8pt>{1}="1" 
  !{a(210) }*+<2.8pt>{2}="2" 
  !{a(330) }*+<2.8pt>{3}="3" 
  !{(0,0);a(90)**{}?(2.4)}*+<2.8pt>{1'}="1'" 
  !{(0,0);a(210)**{}?(2.4)}*+<2.8pt>{2'}="2'" 
  !{(0,0);a(330)**{}?(2.4)}*+<2.8pt>{3'}="3'" 
  "1":"2" 
  "1":"3" 
  "3":"3'" 
  "3":"1'" 
  "1":"1'" 
  "2'":"2" 
  }}}}="q2" 
  !{(2,0)-(0.24,-0.2)}*+{\widetilde{Q}(3)}
  !{(2,0)}*+{
 \vcenter{\xybox{
 \xygraph{!{<0cm,0cm>;<6mm,0cm>:<0cm,6mm>::}
  !{a(90)}*+<2.8pt>{1}="1" 
  !{a(210) }*+<2.8pt>{2}="2" 
  !{a(330) }*+<2.8pt>{3}="3" 
  !{(0,0);a(90)**{}?(2.4)}*+<2.8pt>{1'}="1'" 
  !{(0,0);a(210)**{}?(2.4)}*+<2.8pt>{2'}="2'" 
  !{(0,0);a(330)**{}?(2.4)}*+<2.8pt>{3'}="3'" 
  "1":"2" 
  "1":"3'" 
  "3":"1" 
  "1'":"3" 
  "2'":"2" 
  "3'":"3" 
}}} 
}="q3" 
!{(2,1)-(0.24,-0.2)}*+{\widetilde{Q}(4)}
!{(2,1)}*+{ 
 \vcenter{\xybox{
 \xygraph{!{<0cm,0cm>;<6mm,0cm>:<0cm,6mm>::}
  !{a(90)}*+<2.8pt>{1}="1" 
  !{a(210) }*+<2.8pt>{2}="2" 
  !{a(330) }*+<2.8pt>{3}="3" 
  !{(0,0);a(90)**{}?(2.4)}*+<2.8pt>{1'}="1'" 
  !{(0,0);a(210)**{}?(2.4)}*+<2.8pt>{2'}="2'" 
  !{(0,0);a(330)**{}?(2.4)}*+<2.8pt>{3'}="3'" 
  "1":"3" 
  "2":"1" 
  "3":"2" 
  "1'":"3" 
  "2'":"2" 
  "3'":"1" 
}}}
}="q4" 
!{(1,1)}*+{ 
 \vcenter{\xybox{
 \xygraph{!{<0cm,0cm>;<6mm,0cm>:<0cm,6mm>::}
  !{a(90)}*+<2.8pt>{3}="3" 
  !{a(210) }*+<2.8pt>{2}="2" 
  !{a(330) }*+<2.8pt>{1}="1" 
  !{(0,0);a(90)**{}?(2.4)}*+<2.8pt>{1'}="1'" 
  !{(0,0);a(210)**{}?(2.4)}*+<2.8pt>{2'}="2'" 
  !{(0,0);a(330)**{}?(2.4)}*+<2.8pt>{3'}="3'" 
  "1":"3" 
  "2":"1" 
  "3":"2" 
  "1'":"3" 
  "2'":"2" 
  "3'":"1" 
}}}}="q4a" 
"q0":@{=>}_{\textstyle \mu_{1}}"q1"
"q1":@{=>}^{\textstyle \mu_{2}}"q2"
"q2":@{=>}^{\textstyle \mu_{3}}"q3"
"q3":@{=>}_{\textstyle \mu_{1}}"q4"
"q4":@{=}"q4a"
"q4a":@{=>}_{\textstyle \varphi=(13)}"q0"
}
\end{equation*}
\caption{Reddening mutation loop for $A_{2}^{(1)}$-quiver.} 
 \label{fig:A2(1)-loop}
\end{figure}

\subsection{Square product $A_{3} \square A_{2}$}

As an example of the quivers of square product type (see \cite{KT2014}
for definition), consider
\begin{equation}
 \label{eq:A3xA2-quiver}
 Q=A_{3} \square A_{2}=\vcenter{
 \xymatrix @R=6mm @C=6mm @M=4pt{
 1 \ar[r] & 3  \ar[d] & 5\ar[l]\\
 2  \ar[u]& 4 \ar[l] \ar[r]& 6 \ar[u]}
}.
\end{equation}
One can check that 
\begin{equation}
 \label{eq:A3xA2-revend}
 \seq{m}=(1, 4, 5, 2, 3, 6, 1, 4, 5) 
\end{equation}
is a reddening sequence with the boundary condition
\begin{equation}
\varphi= (12)(34)(56)=
\setlength\arraycolsep{3pt}
\left(
\begin{array}{cccccc}
 1 & 2 & 3 & 4 & 5 & 6 \\
 2 & 1 & 4 & 3 & 6 & 5 \\
\end{array}
\right)\in S_{6}.
\end{equation}
Let $\seq{k}=(k_{1},\dots,k_{9})$ be the $k$-variables corresponding to
mutation sequence \eqref{eq:A3xA2-revend}. The evolution of
$s$-variables along the mutation loop is summarized as follows:
\begin{equation*}
\begin{array}{rclclccc}
 s_1 & \mapsto & s'_1 = k_1-s_1+s_2 & \mapsto & s''_1 =k_7-s'_1+s'_2 &
  = & s_{2}\\
 s_2 & \mapsto & s'_2 = k_4-s_2+s'_1 & && = & s_{1}\\
 s_3 & \mapsto & s'_3 = k_5-s_3+s'_4  & && = & s_{4}\\
 s_4 & \mapsto & s'_4 = k_2-s_4+ s_3 & \mapsto & s''_4 = k_8-s'_4+s'_3 &
  = & s_{3} \\
 s_5 & \mapsto & s'_5 = k_3-s_5+s_6 & \mapsto & s''_5 = k_9-s'_5+s'_6
  & = & s_{6}\\
 s_6 & \mapsto & s'_6 = k_6-s_6+s'_5 & && = & s_{5}\\\\
\end{array}
\end{equation*}
One can express all $s$-variables in terms of $k$-variables:
\begin{equation*}
\setlength\arraycolsep{1pt}
\begin{array}{rlrlrl}
 s_1 = s'_2 & =  \left(k_1+k_4\right)/2,&\qquad
 s_{2} = s''_1 & =   \left(k_4+k_7\right)/2,&\qquad
 s_3 = s''_{4} & =  \left(k_5+k_8\right)/2,\\
 s_{4}= s'_3 & =    \left(k_2+k_5\right)/2,&
 s_5 =s'_6& =   \left(k_3+k_6\right)/2,&
 s_6 =  s''_5& =   \left(k_6+k_9\right)/2,\\
 s'_1 & =   \left(k_1+k_7\right)/2,&
 s'_4 & =   \left(k_2+k_8\right)/2,&
 s'_5 & =   \left(k_3+k_9\right)/2,
\end{array}
\end{equation*}
The $c$-vectors of mutating vertices are
\begin{equation*}
\begin{split}
 &
 \alpha_{1}=(100000),\quad\alpha_{2}=(000100),\quad\alpha_{3}=(000010), 
 \\
 &
 \alpha_{4}=(110000),\quad\alpha_{5}=(001100),\quad\alpha_{6}=(000011),
 \\
 &
 \alpha_{7}=(010000),\quad\alpha_{8}=(001000),\quad\alpha_{9}=(000001).
\end{split}
\end{equation*}
We obtain the partition $q$-series 
\begin{equation*}
 \ZZ(\gamma)=\sum_{\seq{k}\in \N^{9}}
  \frac{q^{\frac{1}{4}\seq{k}^{\top}A\,\seq{k}}}{\prod_{t=1}^{9}(q)_{k_{t}}} 
  y^{\beta(\seq{k})}
\end{equation*}
where
\begin{equation*}
 \beta(\seq{k})=(k_1+k_4,k_4+k_7,k_5+k_8,k_2+k_5,k_3+k_6,k_6+k_9)
\end{equation*}
and $A$ is a symmetric  $9\times 9$ matrix given by
\begin{equation*}
 A=\left(
    \begin{array}{c|c|c}
     A' & A'' & A'' \\
     \hline
      \rule{0pt}{2.5ex} A''& A' &A'' \\
     \hline
      \rule{0pt}{2.5ex} A'' & A''& A' \\
    \end{array}
	       \right),
\qquad 
 A'=
 { \setlength\arraycolsep{4pt}
 \left(\begin{array}{rrr}
  2 & 0 & 0  \\
  0 & 2 & 0  \\
  0 & 0 & 2  \\
	\end{array}\right),}
\qquad 
 A''=
 { \setlength\arraycolsep{2pt}
 \left(\begin{array}{rrr}
  1 & -1 & -1  \\
	-1 & 1 & -1 \\
	-1 & -1 & 1 \\
	\end{array}\right).}
\end{equation*}
\begin{rem}
 The mutation loop \eqref{eq:A3xA2-revend} is different from the one
 considered in our previous work:
 \begin{equation*}
  \gamma' =(Q,\seq{m}',\id)\qquad \seq{m}'=(1, 4, 5, 2, 3, 6).
 \end{equation*}
 (See Theorem 6.1 of \cite{KT2014} and the example therein.) Although
 $\mu_{\seq{m}'}(Q)$ is isomorphic to $Q$, $\seq{m}'$ is not a reddening
 sequence. The sequence $\seq{m}$ contains $\seq{m}'$ as a proper
 subsequence.
\end{rem}

\subsection{Octahedral quiver}

Here is another example of non-alternating quiver --- the octahedral
quiver:
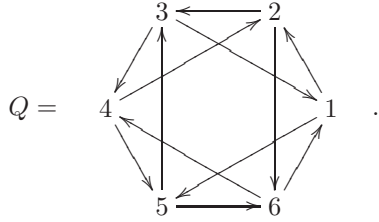
\begin{equation*}
 Q=\quad\vcenter{\xybox{
 \xygraph{!{<0cm,0cm>;<15mm,0cm>:<0cm,15mm>::}
  !{a(0)}*+{1}="1" 
  !{a(60) }*+{2}="2" 
  !{a(120) }*+{3}="3" 
  !{a(180) }*+{4}="4" 
  !{a(240) }*+{5}="5" 
  !{a(300) }*+{6}="6" 
  "1":"2" 
  "2":"3" 
  "3":"4" 
  "4":"5" 
  "5":"6" 
  "6":"1" 
  "3":"1" 
  "4":"2" 
  "5":"3" 
  "6":"4" 
  "1":"5" 
  "2":"6" 
}}}\quad.
\end{equation*}
The mutation sequence 
\begin{equation*}
 \seq{m}=(1, 2, 5, 6, 3, 4, 1, 2, 5, 6, 3, 4)
\end{equation*}
together with the boundary condition
\begin{equation*}
\setlength\arraycolsep{3pt}
 \varphi=
  \begin{pmatrix}
   1 & 2 & 3 & 4 & 5 & 6 \\
   1 & 5 & 6 & 4 & 2 & 3 \\
  \end{pmatrix}
 \in S_{6}.
\end{equation*}
form a reddening, maximal green mutation loop
$\gamma=(Q;\seq{m},\varphi)$ of length $T=12$. 
Indeed, the $c$-matrix of the final quiver $Q(T)$ is given by
\begin{equation*}
\setlength\arraycolsep{2.8pt}
\left(
\begin{array}{rrrrrr}
 -1 & 0 & 0 & 0 & 0 & 0 \\
 0 & 0 & 0 & 0 & -1 & 0 \\
 0 & 0 & 0 & 0 & 0 & -1 \\
 0 & 0 & 0 & -1 & 0 & 0 \\
 0 & -1 & 0 & 0 & 0 & 0 \\
 0 & 0 & -1 & 0 & 0 & 0 \\
\end{array}
\right)
\end{equation*}
which is a (negative of) permutation matrix corresponding to $\varphi$.

Let $\seq{k}=(k_{1},\dots,k_{12})$ be the $k$-variables corresponding to
mutation sequence $\seq{m}$. In this example, every vertex is mutated
twice.  The evolution of $s$-variables along the mutation loop is
summarized as follows:
\begin{equation*}
\renewcommand{\arraystretch}{1.16}
\begin{array}{lllllll}
s_{1} & \mapsto & s'_1=k_1-s_1+s_3+s_6 & \mapsto &
 s''_1=k_7-s'_1+s'_3+s'_6 &=& s_{1}
\\
s_{2} & \mapsto & s'_2=k_2-s_2+s_4 & \mapsto & s''_2=k_8-s'_2+s'_4
&=& s_{5}
\\
s_{3} & \mapsto & s'_3=k_5-s_3+s'_1  & \mapsto & s''_3=k_{11}-s'_3+s''_1 
&=& s_{6} \\
s_{4} & \mapsto & s'_4=k_6-s_4+s'_2+s'_5  & \mapsto &
 s''_4=k_{12}-s'_4 +s''_2+s''_5
&=& s_{4}
\\
s_{5} & \mapsto & s'_5=k_3-s_5+s_4  & \mapsto & s''_5=k_9-s'_5  +s'_4
&=& s_{2}
\\
s_{6} & \mapsto & s'_6=k_4-s_6+s'_1  & \mapsto & s''_6=k_{10}-s'_6+s''_1
&=& s_{3}
\end{array}
\end{equation*}
Solving these, we can express all $s$-variables in terms of $k$-variables:
\begin{align*}
 s_1 & =s''_{1} =  \left(k_1+k_4+k_5+k_7\right)/2, &
 s_2 & =s''_{5}  =   \left(k_2+k_6+k_9\right)/2, \\ 
 s_3 & =s''_{6}  =   \left(k_5+k_7+k_{10}\right)/2, &
 s_4 & =s''_{4} =   \left(k_6+k_8+k_9+k_{12}\right)/2, \\
 s_5 & =s''_{2}  =   \left(k_3+k_6+k_8\right)/2, &
 s_6 & =s''_{3} =   \left(k_4+k_7+k_{11}\right)/2,\\
 s'_1 & =   \left(k_1+k_7+k_{10}+k_{11}\right)/2, &
 s'_2 & =   \left(k_2+k_8+k_{12}\right)/2\\
 s'_3 & =   \left(k_1+k_5+k_{11}\right)/2, &
 s'_4 & =   \left(k_2+k_3+k_6+k_{12}\right)/2\\
 s'_5 & =   \left(k_3+k_9+k_{12}\right)/2, &
 s'_6 & =   \left(k_1+k_4+k_{10}\right)/2.\\
\end{align*}
The partition $q$-series is now given by
\begin{equation}
 \label{eq:Z-octahedron}
 \ZZ(\gamma)=\sum_{\seq{k}\in \N^{12}}
  \frac{q^{\frac{1}{4}\seq{k}^{\top}A\,\seq{k}}}{\prod_{i} (q)_{k_{i}}}
  y^{\beta(\seq{k})},
\end{equation}
where 
\begin{equation*}
\begin{split}
 \beta(\seq{k})=&
 \bigl(
k_1+k_4+k_5+k_7,k_2+k_6+k_9,k_5+k_7+k_{10},\\
 & \qquad k_6+k_8+k_9+k_{12},k_3+k_6+k_8,k_4+k_7+k_{11}\bigr)
 \in \N^{6}
\end{split}
\end{equation*}
and $A$ is the 
$12\times 12$ symmetric matrix of the following form:
\begin{equation*}
 A=\left(
\begin{array}{c|c}
 A'& A'' \\
 \hline
 \rule{0pt}{2.5ex} A''& A' \\
\end{array}
\right),
\end{equation*}
\begin{equation*}
\setlength\arraycolsep{2pt}
  A'=
\left(\begin{array}{rrrrrr}
 2 & -1 & -1 & 1 & 1 & -2 \\
 -1 & 2 & 0 & 0 & 0 & 1 \\
 -1 & 0 & 2 & 0 & 0 & 1 \\
 1 & 0 & 0 & 2 & 0 & -1 \\
 1 & 0 & 0 & 0 & 2 & -1 \\
 -2 & 1 & 1 & -1 & -1 & 2 \\
\end{array}\right),\quad 
  A''=
\left(\begin{array}{rrrrrr}
 2 & -1 & -1 & 1 & 1 & 0 \\
 -1 & 1 & 1 & -1 & -1 & 1 \\
 -1 & 1 & 1 & -1 & -1 & 1 \\
 1 & -1 & -1 & 1 & 1 & -1 \\
 1 & -1 & -1 & 1 & 1 & -1 \\
 0 & 1 & 1 & -1 & -1 & 2 \\
\end{array}\right).
\end{equation*}
By Theorem \ref{thm:main}, the partition $q$-series
\eqref{eq:Z-octahedron} is equal
to $\overline{\EE(Q;\seq{m})}$, where 
\begin{equation*}
\begin{split}
\EE(Q;\seq{m})&= 
 \EE(y^{(1 0 0 0 0 0)})
 \EE(y^{(0 1 0 0 0 0)})
 \EE(y^{(0 0 0 0 1 0)})
 \EE(y^{(1 0 0 0 0 1)})
 \EE(y^{(1 0 1 0 0 0)})
 \EE(y^{(0 1 0 1 1 0)}) 
 \\
 &\times 
 \EE(y^{(1 0 1 0 0 1)}) 
 \EE(y^{(0 0 0 1 1 0)}) 
 \EE(y^{(0 1 0 1 0 0)}) 
 \EE(y^{(0 0 1 0 0 0)}) 
 \EE(y^{(0 0 0 0 0 1)}) 
 \EE(y^{(0 0 0 1 0 0)}).
\end{split}
\end{equation*}
is the product of quantum dilogarithms,
$\EE(y^\alpha)=\EE(y^{\alpha};q)$.

\subsection{Alternating quivers}

A vertex $i$ of a quiver is a \emph{source} (respectively, a
\emph{sink}) if there are no arrows $\alpha$ with target $i$
(respectively, with source $i$).  A quiver is \emph{alternating} if each
of its vertices is a source or a sink.  Denote by $Q_{0}^{+}$
($Q_{0}^{-}$) the set of all sources (sinks) of the alternating quiver
$Q$, respectively. Since $Q_{0}=Q_{0}^{+}\sqcup Q_{0}^{-}$, the
underlying graph $\underline{Q}$, a graph obtained by forgetting the
orientation of arrows, is bipartite.

For an alternating quiver $Q$, there is a simple recipe for constructing
a reddening sequence/loop.

\begin{prop}
 \label{prop:alter-max-green} Suppose $Q$ is an alternating quiver, and
 $\seq{m}_{\pm}$ be arbitrary permutations of $Q_{0}^{\pm}$,
 respectively. Let $\seq{m}=\seq{m}_{+}\seq{m}_{-}$ be their
 concatenation, considered as a mutation sequence
 of length $n=|Q_{0}|$.  Then, the
 $c$-vectors of $Q(t)$ are given by
  \begin{equation}
   \label{eq:alt-quiv-c-evol}
   c_{i}(t)=
    \begin{cases}
   e_{i}& \text{if~} i\not\in \{m_{1},\dots,m_{t}\},
     \\
   -e_{i}& \text{if~} i\in \{m_{1},\dots,m_{t}\},
    \end{cases}
    \qquad (0\leq t\leq n).
  \end{equation}
 In particular, the sequence $\seq{m}$ is maximal green and
 $\gamma=(Q;\seq{m}_{+}\seq{m}_{-},\id)$ is a reddening mutation
 loop.
\end{prop}
\begin{proof}
 First note that $\seq{m}$ is a \emph{source sequence}, that is, each
 mutating vertex $m_{t}$ is a source of $Q(t-1)$ for all $1\leq t\leq
 n$. To see this, it is helpful to consider $\seq{m}_{+}$ and
 $\seq{m}_{-}$ separately. The claim is clear for the mutation sequence
 $\seq{m}_{+}$ applied on $Q$. When the mutation sequence $\seq{m}_{+}$
 is over, we have $\mu_{\seq{m}_{+}}(Q)=Q^{op}$; here $Q^{op}$ is the
 quiver obtained by reversing all the arrows in $Q$. Now all the
 vertices in $\seq{m}_{-}$ are sources of $\mu_{+}(Q)=Q^{op}$, so
 $\seq{m}_{-}$ is also a source sequence.

 Since only source vertices are mutated, mutation rules 1) and 3) are
 never used; mutations change only the orientations of arrows. The
 underlying graph $\underline{Q}$ remains the same.

 Let $M(t):=\{m_{1},\dots,m_{t}\}\subset Q_{0}$ be the set of mutated
 vertices during the first $t$ mutations. We prove
 \eqref{eq:alt-quiv-c-evol} by induction on $t$.  The claim holds for
 $t=0$, since $M(0)=\emptyset$ and $c_{i}(0)=e_{i}$ for all $i$. Suppose
 the claim is true for $0,1,\dots,t-1$. Then the mutation
 $\mu_{m_{t}}:Q(t{-}1)\to Q(t)$ is green because $m_{t}\not\in
 \{m_{1},\dots,m_{t-1}\}$ and thus $c_{m_{t}}(t{-}1)=e_{m_{t}}\in \N^{n}$
 by induction hypothesis.  Moreover, $Q(t{-}1)_{i,m_{t}}=0$ since $m_{t}$
 is a source of $Q(t{-}1)$, as we have seen above. Thus by
 \eqref{eq:c-vec-change}, the $c$-vectors change as
 \begin{equation*}
  c_{i}(t)=\begin{cases}
	    -c_{i}(t{-}1) & \text{if~} i=m_{t}, \\
	    c_{i}(t{-}1) & \text{if~} i\neq m_{t}. 
	   \end{cases}
 \end{equation*}
 With $M(t)=M(t{-}1)\sqcup \{m_{t}\}$, this shows that the claim is also
 true for $t$.  The rest of the proposition follows immediately from
 \eqref{eq:alt-quiv-c-evol}.
\end{proof}

Let us compute $\ZZ(\gamma)$ for the reddening loop
$\gamma=(Q;\seq{m}_{+}\seq{m}_{-},\id)$. Note that the sequence
$\seq{m}=(m_{1},\dots,m_{n})$ is a permutation of $(1,\dots,n)$. Every
vertex $i$ is mutated exactly once, and the initial and final
$s$-variables $s_{i}$, $s'_{i}$ are identified by the boundary condition
$\varphi=\id$.  As we will soon see, it is convenient to label
$k$-variables not by the mutation time but by the vertex label. From now
on, $k_{i}$ will denote the $k$-variable associated with the mutation at
vertex $i$, rather than $i$-th mutation.

To compute the weight for $\gamma$, it suffices to know the underlying
graph $\underline{Q}$, because we can recover arrow orientations from
the fact that ``every mutation occurs at a source''. All the information
of $\underline{Q}$ is encoded in the \emph{generalized Cartan matrix}
 \begin{equation}
  \label{eq:gen-Cartan}
   (C)_{ij}=
  \begin{cases}
   2 & \text{if~} i=j, \\
   -(Q_{ij}+Q_{ji}) & \text{if~} i\neq j.
  \end{cases}
 \end{equation}

Before stating the general result for $\ZZ(\gamma)$ (Theorem
\ref{thm:qdilog-char}), let us take an example --- an alternating quiver
of affine $D_{5}$ type:
\begin{equation}
 \label{eq:affineA5}
 Q=
\vcenter{\xygraph{!{<0cm,0cm>;<16pt,0mm>:<0cm,16pt>::}
!{(0,2)}*+<2pt>{1}="1"
!{(0,0)}*+<2pt>{2}="2"
!{(0.9,1)}*+<2pt>{3}="3"
!{(2.1,1)}*+<2pt>{4}="4"
!{(3,2)}*+<2pt>{5}="5"
!{(3,0)}*+<2pt>{6}="6"
"1":"3"  
"2":"3"  
"4":"3"  
"4":"5"  
"4":"6"  
}}
\;.
\end{equation}
The generalized Cartan matrix of $Q$ is given by
\begin{equation}
 \label{eq:GCM-D5(1)}
 \arraycolsep=2.8pt
   C=\begin{pmatrix}
      2 & 0 & -1 & 0 & 0 & 0 \\
      0 & 2 & -1 & 0 & 0 & 0 \\
      -1 & -1 & 2 & -1 & 0 & 0 \\
      0 & 0 & -1 & 2 & -1 & -1 \\
      0 & 0 & 0 & -1 & 2 & 0 \\
      0 & 0 & 0 & -1 & 0 & 2 \\
     \end{pmatrix}.
\end{equation}
Put $\seq{m}_{+}=(1,2,4)$ and $\seq{m}_{-}=(3,5,6)$. By Proposition
\ref{prop:alter-max-green}, the mutation sequence
\begin{equation*}
 \seq{m}=\seq{m}_{+}\seq{m}_{-}=(1,2,4, 3,5,6)
\end{equation*}
is maximal green, reddening sequence with the boundary condition
$\varphi=\id$ (see Figure \ref{fig:affineD5}).

\begin{figure}[t]
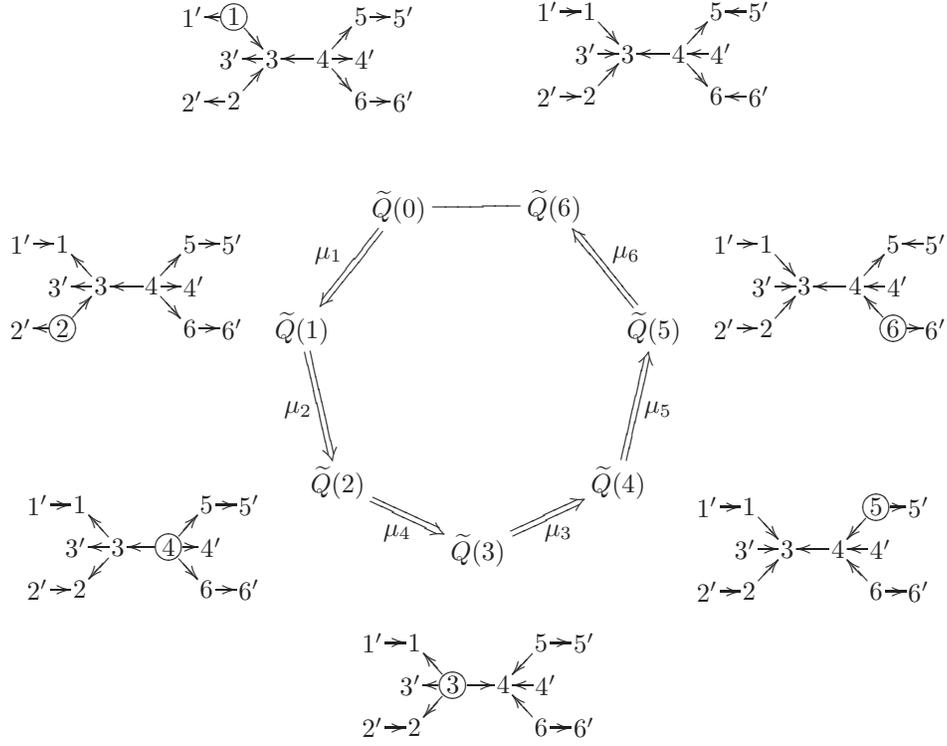

\begin{equation*}
 \xygraph{!{<0cm,0cm>;<12mm,0cm>:<0cm,12mm>::}
 !{(0,0);a(116)**{}?(2.0)}*+<5pt>{\widetilde{Q}(0)}="0"
 !{(0,0);a(120)**{}?(4)}*+<2pt>{
  \vcenter{\xybox{\xygraph{!{<0cm,0cm>;<16pt,0mm>:<0cm,16pt>::}
!{(0,2)}*+[o][F-]{1}="1"
!{(0,0)}*+<2pt>{2}="2"
!{(0.9,1)}*+<2pt>{3}="3"
!{(2.1,1)}*+<2pt>{4}="4"
!{(3,2)}*+<2pt>{5}="5"
!{(3,0)}*+<2pt>{6}="6"
!{(-1,2)}*+<2pt>{1'}="1'"
!{(-1,0)}*+<2pt>{2'}="2'"
!{(-0.1,1)}*+<2pt>{3'}="3'"
!{(3.1,1)}*+<2pt>{4'}="4'"
!{(4,2)}*+<2pt>{5'}="5'"
!{(4,0)}*+<2pt>{6'}="6'"
"1":"3"  
"2":"3"  
"4":"3"  
"4":"5"  
"4":"6"  
"1":"1'"  
"2":"2'"  
"3":"3'"  
"4":"4'"  
"5":"5'"  
"6":"6'"  
}}}}
 !{(0,0);a(167)**{}?(2.0)}*+<5pt>{\widetilde{Q}(1)}="1"
 !{(0,0);a(167)**{}?(4)}*+<2pt>{
  \vcenter{\xybox{\xygraph{!{<0cm,0cm>;<16pt,0mm>:<0cm,16pt>::}
!{(0,2)}*+<2pt>{1}="1"
!{(0,0)}*+[o][F-]{2}="2"
!{(0.9,1)}*+<2pt>{3}="3"
!{(2.1,1)}*+<2pt>{4}="4"
!{(3,2)}*+<2pt>{5}="5"
!{(3,0)}*+<2pt>{6}="6"
!{(-1,2)}*+<2pt>{1'}="1'"
!{(-1,0)}*+<2pt>{2'}="2'"
!{(-0.1,1)}*+<2pt>{3'}="3'"
!{(3.1,1)}*+<2pt>{4'}="4'"
!{(4,2)}*+<2pt>{5'}="5'"
!{(4,0)}*+<2pt>{6'}="6'"
"3":"1"
"2":"3"  
"4":"3"  
"4":"5"  
"4":"6"  
"1'":"1"  
"2":"2'"  
"3":"3'"
"4":"4'"  
"5":"5'"  
"6":"6'"  
}}}}
 !{(0,0);a(219)**{}?(2.0)}*+<5pt>{\widetilde{Q}(2)}="2"
 !{(0,0);a(208)**{}?(4.2)}*+<2pt>{
  \vcenter{\xybox{\xygraph{!{<0cm,0cm>;<16pt,0mm>:<0cm,16pt>::}
!{(0,2)}*+<2pt>{1}="1"
!{(0,0)}*+<2pt>{2}="2"
!{(0.9,1)}*+<2pt>{3}="3"
!{(2.1,1)}*+[o][F-]{4}="4"
!{(3,2)}*+<2pt>{5}="5"
!{(3,0)}*+<2pt>{6}="6"
!{(-1,2)}*+<2pt>{1'}="1'"
!{(-1,0)}*+<2pt>{2'}="2'"
!{(-0.1,1)}*+<2pt>{3'}="3'"
!{(3.1,1)}*+<2pt>{4'}="4'"
!{(4,2)}*+<2pt>{5'}="5'"
!{(4,0)}*+<2pt>{6'}="6'"
"3":"1"
"3":"2"  
"4":"3"  
"4":"5"  
"4":"6"  
"1'":"1"  
"2'":"2"  
"3":"3'"
"4":"4'"  
"5":"5'"  
"6":"6'"  
}}}}
 !{(0,0);a(270)**{}?(2.0)}*+<5pt>{\widetilde{Q}(3)}="3"
 !{(0,0);a(270)**{}?(3.5)}*+<2pt>{
  \vcenter{\xybox{\xygraph{!{<0cm,0cm>;<16pt,0mm>:<0cm,16pt>::}
!{(0,2)}*+<2pt>{1}="1"
!{(0,0)}*+<2pt>{2}="2"
!{(0.9,1)}*+[o][F-]{3}="3"
!{(2.1,1)}*+<2pt>{4}="4"
!{(3,2)}*+<2pt>{5}="5"
!{(3,0)}*+<2pt>{6}="6"
!{(-1,2)}*+<2pt>{1'}="1'"
!{(-1,0)}*+<2pt>{2'}="2'"
!{(-0.1,1)}*+<2pt>{3'}="3'"
!{(3.1,1)}*+<2pt>{4'}="4'"
!{(4,2)}*+<2pt>{5'}="5'"
!{(4,0)}*+<2pt>{6'}="6'"
"3":"1"
"3":"2"  
"3"  :"4"
"5"  :"4"
"6"  :"4"
"1'":"1"  
"2'":"2"  
"3":"3'"
"4'":"4"  
"5":"5'"  
"6":"6'"  
}}}}
 !{(0,0);a(321)**{}?(2.0)}*+<5pt>{\widetilde{Q}(4)}="4"
 !{(0,0);a(332)**{}?(4.2)}*+<2pt>{
  \vcenter{\xybox{\xygraph{!{<0cm,0cm>;<16pt,0mm>:<0cm,16pt>::}
!{(0,2)}*+<2pt>{1}="1"
!{(0,0)}*+<2pt>{2}="2"
!{(0.9,1)}*+<2pt>{3}="3"
!{(2.1,1)}*+<2pt>{4}="4"
!{(3,2)}*+[o][F-]{5}="5"
!{(3,0)}*+<2pt>{6}="6"
!{(-1,2)}*+<2pt>{1'}="1'"
!{(-1,0)}*+<2pt>{2'}="2'"
!{(-0.1,1)}*+<2pt>{3'}="3'"
!{(3.1,1)}*+<2pt>{4'}="4'"
!{(4,2)}*+<2pt>{5'}="5'"
!{(4,0)}*+<2pt>{6'}="6'"
"1":"3"
"2":"3"  
"4"  :"3"
"5"  :"4"
"6"  :"4"
"1'":"1"  
"2'":"2"  
"3'":"3"
"4'":"4"  
"5":"5'"  
"6":"6'"  
}}}}
 !{(0,0);a(13)**{}?(2.0)}*+<5pt>{\widetilde{Q}(5)}="5"
 !{(0,0);a(13)**{}?(4)}*+<2pt>{
  \vcenter{\xybox{\xygraph{!{<0cm,0cm>;<16pt,0mm>:<0cm,16pt>::}
!{(0,2)}*+<2pt>{1}="1"
!{(0,0)}*+<2pt>{2}="2"
!{(0.9,1)}*+<2pt>{3}="3"
!{(2.1,1)}*+<2pt>{4}="4"
!{(3,2)}*+<2pt>{5}="5"
!{(3,0)}*+[o][F-]{6}="6"
!{(-1,2)}*+<2pt>{1'}="1'"
!{(-1,0)}*+<2pt>{2'}="2'"
!{(-0.1,1)}*+<2pt>{3'}="3'"
!{(3.1,1)}*+<2pt>{4'}="4'"
!{(4,2)}*+<2pt>{5'}="5'"
!{(4,0)}*+<2pt>{6'}="6'"
"1":"3"
"2":"3"  
"4"  :"3"
"4"  :"5"
"6"  :"4"
"1'":"1"  
"2'":"2"  
"3'":"3"
"4'":"4"  
"5'":"5"  
"6":"6'"  
}}}}
 !{(0,0);a(65)**{}?(2.0)}*+<5pt>{\widetilde{Q}(6)}="6"
 !{(0,0);a(61)**{}?(4)}*+<2pt>{
  \vcenter{\xybox{\xygraph{!{<0cm,0cm>;<16pt,0mm>:<0cm,16pt>::}
!{(0,2)}*+<2pt>{1}="1"
!{(0,0)}*+<2pt>{2}="2"
!{(0.9,1)}*+<2pt>{3}="3"
!{(2.1,1)}*+<2pt>{4}="4"
!{(3,2)}*+<2pt>{5}="5"
!{(3,0)}*+<2pt>{6}="6"
!{(-1,2)}*+<2pt>{1'}="1'"
!{(-1,0)}*+<2pt>{2'}="2'"
!{(-0.1,1)}*+<2pt>{3'}="3'"
!{(3.1,1)}*+<2pt>{4'}="4'"
!{(4,2)}*+<2pt>{5'}="5'"
!{(4,0)}*+<2pt>{6'}="6'"
"1":"3"
"2":"3"  
"4"  :"3"
"4"  :"5"
"4"  :"6"
"1'":"1"  
"2'":"2"  
"3'":"3"
"4'":"4"  
"5'":"5"  
"6'":"6"  
}}}}
 "0":@{=>}_{\textstyle \mu_{1}}"1"
 "1":@{=>}_{\textstyle \mu_{2}}"2"
 "2":@{=>}_{\textstyle \mu_{4}}"3"
 "3":@{=>}_{\textstyle \mu_{3}}"4"
 "4":@{=>}_{\textstyle \mu_{5}}"5"
 "5":@{=>}_{\textstyle \mu_{6}}"6"
 "6":@{-}"0"
}
\end{equation*}
\caption{Reddening mutation loop for affine $D_{5}$ alternating
 quiver. Mutating vertices are marked with circles.}  \label{fig:affineD5}
\end{figure}

The $s$-variables change as follows:
\begin{equation*}
 \renewcommand{\arraystretch}{1.1}
 \begin{array}{c|c|c|c|c|c|c}
  & 1 & 2 & 3 & 4 & 5 & 6 \\
  \hline
 Q(0) & s_1 & s_2 & s_3 & s_4 & s_5 & s_6 \\
 Q(1) & s'_1{=}k_1{-}s_1 & s_2 & s_3 & s_4 & s_5 & s_6 \\
 Q(2) & s'_1{=}k_1{-}s_1 & s'_2{=}k_2{-}s_2 & s_3 & s_4 & s_5 & s_6 \\
 Q(3) & s'_1{=}k_1{-}s_1 & s'_2{=}k_2{-}s_2 & s_3 
  & s'_4{=}k_4{-}s_4 & s_5 & s_6 \\
 Q(4) & s'_1{=}k_1{-}s_1 & s'_2{=}k_2{-}s_2 & s'_3{=}k_3{-}s_3
  & s'_4{=}k_4{-}s_4 & s_5 & s_6 \\
 Q(5) & s'_1{=}k_1{-}s_1 & s'_2{=}k_2{-}s_2 & s'_3{=}k_3{-}s_3 
  & s'_4{=}k_4{-}s_4 & s'_5{=}k_5{-}s_5 & s_6 \\
 Q(6) & s'_1{=}k_1{-}s_1 & s'_2{=}k_2{-}s_2 & s'_3{=}k_3{-}s_3 
  & s'_4{=}k_4{-}s_4 & s'_5{=}k_5{-}s_5 & s'_6{=}k_6{-}s_6 \\
\end{array}
\end{equation*}

The boundary condition
$\varphi=\id$ imposes $s_{i}= s_{i}'=k_i-s_i$ for all $i$, so we have
\begin{equation*}
 s_i = s_{i}'= \frac{1}{2} k_{i} \qquad (i=1,\dots,6).
\end{equation*}
The $k^{\vee}$-variables (also labeled by mutated vertices) are then
given by
\begin{equation*}
\renewcommand{\arraystretch}{1.2}
\arraycolsep=2pt
\begin{array}{rclcl}
 k^{\vee}_1 &=& s_{1}+s'_{1}-s_3 &=& k_1-\frac{k_3}{2}, \\
 k^{\vee}_2 &=& s_{2}+s'_{2}-s_3 &=& k_2-\frac{k_3}{2}, \\
 k^{\vee}_4 &=& s_{4}+s'_{4}-(s_3+s_5+s_6)
  &=& -\frac{k_3}{2}+k_4-\frac{k_5}{2}-\frac{k_6}{2}, \\
 k^{\vee}_3 &=& s_{3}+s'_{3}-(s'_1+s'_2+s'_4) &=& 
  -\frac{k_1}{2}-\frac{k_2}{2}+k_3-\frac{k_4}{2}, \\
 k^{\vee}_5 &=& s_{5}+s'_{5}-s'_4 &=& -\frac{k_4}{2}+k_5, \\
 k^{\vee}_6 &=& s_{6}+s'_{6}-s'_4 &=& -\frac{k_4}{2}+k_6. \\
\end{array}
\end{equation*}
Thus the the weight of the mutation now reads
\begin{equation*}
 W(\seq{m})=\frac{q^{\frac{1}{2}\sum_{i}k_{i}k^{\vee}_{i}}}
{\prod_{i=1}^{6}(q)_{k_{i}}}
 =
\frac{q^{\frac{1}{2}\left(k_1^2+k_2^2+k_3^2+k_4^2+k_5^2+k_6^2- k_1 k_3-k_2 k_3-k_3
		     k_4-k_4 k_5-k_4 k_6\right)}}{
\prod_{i=1}^{6}(q)_{k_{i}}}.
\end{equation*}
Every mutating vertex $m_{t}$ is green with $c$-vector
$\alpha_{t}=e_{m_{t}}$. The $\N^{n}$-grading of the mutation sequence
$\seq{m}$ is then
\begin{equation*}
 \sum_{t=1}^{6} k_{m_{t}} e_{m_{t}} =
 \sum_{i=1}^{6} k_{i} e_{i} =
 \left(k_1,k_2,k_3,k_4,k_5,k_6\right) =\seq{k}\in \N^{6}.
\end{equation*} 
Combining all these, we obtain a neat expression for the partition
$q$-series:
\begin{equation}
 \label{eq:Z-D5(1)}
 \ZZ(\gamma)=\sum_{\seq{k}\in \N^{6}}
  \frac{q^{\frac{1}{4}\seq{k}^{\top}C\seq{k}}}{\prod_{i=1}^{6} (q)_{k_i}}
  y^{\seq{k}},
\end{equation}
where $C$ is nothing but the generalized Cartan matrix \eqref{eq:GCM-D5(1)}.  

In fact, this generalize to all alternating quivers:
\begin{theorem}
 \label{thm:qdilog-char} Suppose $Q$ is an alternating quiver, and
 $\gamma=(Q;\seq{m},\id)$ be the reddening mutation loop
 constructed as in Proposition \ref{prop:alter-max-green}. Let
 $\seq{k}=(k_{1},\dots,k_{n})$ be the vector of $k$-variables indexed by
 the vertices. Then the partition $q$-series is given by
 \begin{equation}
  \label{eq:qdilog-char}
  \ZZ(\gamma)= \overline{ \EE(\seq{m};q)}
  =\sum_{\seq{k}\in \N^{n}} 
  \frac{q^{\frac{1}{4} \seq{k}^{\top}C\,\seq{k}}}{\prod_i (q)_{k_{i}}}
  y^{\seq{k}},
 \end{equation}
 where $C$ is the generalized Cartan matrix of $Q$ given in
 \eqref{eq:gen-Cartan}.
\end{theorem}

\begin{proof}
 From Proposition \ref{prop:alter-max-green}, $\varepsilon_{t}=1$ and
 $c_{m_{t}}(t-1)=e_{m_{t}}$ for all mutation time $1\leq t\leq n$. Thus
 we have $\alpha_{t}=e_{m_{t}}$ in \eqref{eq:alpha-def}.  The
 $\N^{n}$-grading is therefore given by $\sum_{t=1}^{n} k_{m_{t}}
 e_{m_{t}}=\sum_{i=1}^{n} k_{i} e_{i} =\seq{k}$.

 Consider a mutation at vertex $i$. As we have seen, $i$ is a
 source and there is no arrow ending on $i$. The initial (= final)
 $s$-variable and the $k$-variable are thus related as $2s_{i}=k_{i}$, so we
 have
 \begin{equation}
  s_{i}=\frac{k_{i}}{2}\qquad (1\leq i\leq n).
 \end{equation}
 The $k^{\vee}$-variables are then expressed as 
 \begin{align}
  \label{eq:alt-kv}
  k^{\vee}_{i} 
  &= 2 s_{i} - \sum_{i\to j} s_{j} 
  = 2 s_{i} - \sum_{i\sim j} s_{j}   = k_{i} - \frac{1}{2} \sum_{i\sim
  j} k_{j}.
 \end{align}
 Here $i\sim j$ means the vertices $i$ and $j$ are adjacent in the
 underlying graph $\underline{Q}$. Using the generalized Cartan matrix
 \eqref{eq:gen-Cartan}, the relation \eqref{eq:alt-kv} is concisely
 written as
 \begin{equation}
  \seq{k}^{\vee} = \frac{1}{2} C\seq{k}.
 \end{equation}
 Thus the partition $q$-series is given by
 \begin{equation}
  \label{eq:alt-zz}
  \ZZ(\gamma)  =
   \sum_{k_1,\dots,k_{n} \geq 0}\biggl(
   \prod_{i=1}^{n}
  \frac{q^{\frac{1}{2}k_{i}k^{\vee}_{i}}}{ (q)_{k_{i}}}
  \biggr)
  y^{\seq{k}}
  =
  \sum_{\seq{k}\in \N^{n}}
  \frac{q^{\frac{1}{4} \seq{k}^{\top}C\,\seq{k}}}{\prod_i (q)_{k_{i}}}
  y^{\seq{k}}.
 \end{equation}
 The equality $ \ZZ(\gamma)=\overline{\EE(\seq{m};q)}$ follows from
 Theorem \ref{thm:main}.
\end{proof}

\begin{rem}
 In our previous work, we computed the partition $q$-series for the
 product of Dynkin quivers and observed that they are fermionic
 character formulas of certain conformal field theories (Theorem 6.1 of
 \cite{KT2014}). The case considered here are different from those
 because (i) $Q\square Q'$ is not alternating in general, and (ii) the
 sequences given in \cite{KT2014} are not reddening. However,
 $Q\square Q'=X \square A_{1}$ with $C_{Q'}=(2)$ are exceptional cases
 to which Theorem \ref{thm:main} is applicable.
\end{rem}


\appendix

\section{Some identities related with quantum dilogarithm}

\begin{prop}
 \label{prop:EqEqinv}
 \begin{equation}
  \label{eq:EqEqinv}
  \EE(y;q)\EE(y;q^{-1})=1
 \end{equation}
\end{prop}
\begin{proof} This follows from, for example by 
 exchanging $q\leftrightarrow q^{-1}$ in the expression
 \eqref{eq:E=exp()}.
\end{proof}
\begin{cor} 
 \label{prop:EEinv-cor}
 \begin{equation}
  \label{eq:EEinv-cor} \sum_{\substack{r,s\geq 0\\r+s=n}}
  \frac{q^{\frac{1}{2}r^{2}}}{(q)_{r}}
  \frac{q^{-\frac{1}{2}s^{2}}}{(q^{-1})_{s}}=\delta_{n,0}
  \qquad (n=0,1,2,\dots).
 \end{equation} 
\end{cor}
\begin{proof} This is proved by expanding \eqref{eq:EqEqinv} as a series
 in $y$, and taking the coefficient of $y^{n}$. An alternative proof
 goes as follows.  We begin by the $q$-binomial formula
 \begin{equation}
  \label{eq:q-binom}
   \prod_{k=0}^{n-1} (1+q^kx)=\sum_{r=0}^n q^{\frac{r(r-1)}{2}} 
   \frac{(q)_{n}}{(q)_{r}(q)_{n-r}}
   x^r.
 \end{equation}
 The right hand side of \eqref{eq:q-binom} can be written as 
 \begin{equation*}
  \sum_{\substack{r,s\geq 0\\r+s=n}} q^{\frac{r(r-1)}{2}}
   \frac{(q)_{n}}{(q)_{r}(q)_{s}} x^{r} =\sum_{\substack{r,s\geq
   0\\r+s=n}} \frac{(q)_{n}}{(q)_{r}(q^{-1})_{s}}
   q^{\frac{r(r-1)}{2}-\frac{s(s+1)}{2}} (-1)^{s} x^{r}.
 \end{equation*}
 By putting $x=-1$ into \eqref{eq:q-binom}, we have
 \begin{equation}
  \label{eq:LHS} \prod_{k=0}^{n-1} (1-q^k) =(-1)^{n} q^{-\frac{n}{2}}
   (q)_{n}\sum_{\substack{r,s\geq 0\\r+s=n}} \frac{1}
   {(q)_{r}(q^{-1})_{s}} q^{\frac{r^{2}-s^{2}}{2}}.
 \end{equation}
 The left hand side of \eqref{eq:LHS} is 1 if $n=0$, and 0 otherwise.
\end{proof}


\end{document}